\documentclass[11pt]{article}
\usepackage{fullpage}
\usepackage[T1]{fontenc}
\usepackage{graphicx}
\usepackage{tikz}
\usepackage{authblk}

\usepackage{amsmath,amssymb}

\usepackage{amsthm}
\usepackage{nicematrix}
\usepackage{subcaption}
\usepackage[linewidth=1pt]{mdframed}

\usepackage[ruled,vlined]{algorithm2e}
\usepackage{tcolorbox} 

\usepackage{enumitem}
\usepackage{dsfont}
\usepackage{bm}
\usepackage{enumitem}
\usepackage{calc}

\usepackage{pgfmath}
\usepackage{pgffor}
\usetikzlibrary[patterns]
\usetikzlibrary{calc}
\usetikzlibrary{matrix, positioning}
\usepackage{pgfplots}
\usepgfplotslibrary{fillbetween}
\usetikzlibrary{intersections, patterns}
\pgfplotsset{width=8cm,compat=1.9}

\usepackage{hyperref}
\usepackage[capitalise]{cleveref}
\newtheorem{theorem}{Theorem}[section]
\newtheorem{definition}[theorem]{Definition}
\newtheorem{lemma}[theorem]{Lemma}
\newtheorem{corollary}[theorem]{Corollary}
\newtheorem{claim}[theorem]{Claim}
\newtheorem{conjecture}[theorem]{Conjecture}
\crefname{conjecture}{Conjecture}{Conjectures}
\crefname{claim}{Claim}{Claims}

\newcommand{\abs}[1]{\left\lvert#1\right\rvert}

\newtheorem{fact}{Fact}[section]

\newcommand{\bSigma}{\bm{\Sigma}}

\newcommand{\iid}{\overset{\text{iid}}{\sim}}

\DeclareMathOperator{\E}{\mathbb{E}}

\def\ddefloop#1{\ifx\ddefloop#1\else\ddef{#1}\expandafter\ddefloop\fi}
\def\ddef#1{\expandafter\def\csname bb#1\endcsname{\ensuremath{\mathbb{#1}}}}
\ddefloop ABCDEFGHIJKLMNOPQRSTUVWXYZ\ddefloop

\def\ddef#1{\expandafter\def\csname sf#1\endcsname{\ensuremath{\mathsf{#1}}}}
\ddefloop ABCDEFGHIJKLMNOPQRSTUVWXYZ\ddefloop

\def\ddef#1{\expandafter\def\csname cal#1\endcsname{\ensuremath{\mathcal{#1}}}}
\ddefloop ABCDEFGHIJKLMNOPQRSTUVWXYZ\ddefloop

\def\ddef#1{\expandafter\def\csname vec#1\endcsname{\ensuremath{\mathbf{#1}}}}
\ddefloop abcdefghijklmnopqrstuvwxyz\ddefloop

\def\ddef#1{\expandafter\def\csname mat#1\endcsname{\ensuremath{\mathbf{#1}}}}
\ddefloop ABCDEFGHIJKLMNOPQRSTUVWXYZ\ddefloop
\newcommand{\sign}{\textrm{sign}}

\newenvironment{pseudocode}{
  \par\smallskip
  \begin{center}
  \begin{tcolorbox}[
    width=\linewidth,
    colback=white,
    colframe=black,
    boxrule=0.5pt,
    arc=0pt,
    left=6pt,
    right=6pt,
    top=2pt,
    bottom=2pt,
    boxsep=3pt]
  \small
  \begin{tabbing}
  \qquad\=\qquad\=\qquad\=\qquad\=\qquad\=\qquad\=\qquad\= \kill}{
  \end{tabbing}
  \end{tcolorbox}
  \end{center}
  \smallskip}

\title{Collision Resistance of Single-Layer Neural Nets}

    \author[1]{Marco Benedetti}
    \author[2]{Andrej Bogdanov\thanks{Part of the work was done while visiting Bocconi University, supported by European Research Council (ERC) under the EU’s Horizon 2020 research and innovation programme (Grant agreement No. 101019547).}}
    \author[1]{Enrico M. Malatesta}
    \author[1]{Marc Mézard}
    \author[1]{\\ Gianmarco Perrupato}
    \author[1]{Alon Rosen\thanks{Work supported by European Research Council (ERC) under the EU’s Horizon 2020 research and innovation programme (Grant agreement No. 101019547) and Cariplo CRYPTONOMEX grant.}}
    \author[]{Nikolaj I. Schwartzbach\thanks{Most of the work was done while at Bocconi University, supported by European Research Council (ERC) under the EU’s Horizon 2020 research and innovation programme (Grant agreement No. 101019547).}}
    \author[1]{Riccardo Zecchina}
    \affil[1]{Bocconi University}
    \affil[2]{University of Ottawa}

\begin{document}

\maketitle

\begin{abstract}
We initiate the study of the algorithmic complexity of finding collisions in single-layer binary neural networks. Given a random matrix $\matA \in \mathbb{R}^{m\times n}$, an input $\vecx \in \{-1,1\}^n$ is mapped to a binary output vector $\varphi(\matA\vecx)\in \{-1,1\}^m$, where $\varphi$ is an activation function with constant behavior on $[\kappa, \infty)$ for some threshold $\kappa \geq 0$.


We identify the threshold scale $\kappa=\Theta(1/\sqrt{\alpha})$, where $\alpha=m/n$, as separating two complementary phenomena. When $\kappa \ll 1/\sqrt{\alpha}$, we give a simple online algorithm that efficiently produces extensive collisions. When $\kappa \gg 1/\sqrt{\alpha}$, for a natural \emph{randomized} non-periodic activation and suitable oscillation complexity, we prove that the extensive-collision space exhibits an overlap gap property (OGP), yielding an exponential lower bound against online algorithms.

Ours is the first work to use the overlap gap property as a rigorous criterion for collision resistance. The key difference between collision finding and average-case search is that collision finding has a new ``worst-case'' aspect: the collision finder has full control over the choice of colliding pairs. Our lower bound is proved in the online model; extending such guarantees to broader classes of algorithms, including spectral, algebraic, lattice-based, or quantum methods, remains an open direction.
\end{abstract}


\pagestyle{plain}
\pagenumbering{arabic}

\section{Introduction}

Is it feasible to find two different inputs that map to the same output in a neural net? Finding such collisions is sometimes easy~\cite{LiZM19,OzbulakEtAlSciRep26}, but the question is whether there exist combinations of input domains and activation functions where collision finding becomes hard. 

We consider a standard setup of a single-layer neural network with $n$ binary inputs $\vecx\in \{-1,1\}^n$ and $m$ binary outputs~\cite{rosenblatt1958perceptron,Rosenblatt239697,mcculloch1943logical}. The outputs are obtained by multiplying $\vecx$ with a weight matrix $\matA\in \mathbb{R}^{m\times n}$, followed by applying an activation map $\varphi:\bbR^m\to\{-1,1\}^m$:
\[
f(\vecx)=\varphi(\matA\vecx)\in \{-1,1\}^m\,.
\]

A \emph{collision} is a pair $\vecx\neq\vecx'$ such that $\varphi(\matA\vecx)=\varphi(\matA\vecx')$. We focus on the shrinking regime $m<n$, where collisions must exist by the pigeonhole principle. The question remains interesting even when $m\ge n$ as long as collisions persist. We restrict to discrete inputs since continuous domains permit trivial linear-algebraic collisions via the kernel of $\matA$. We also assume that $\matA\sim \calN(0,1/n)^{m\times n}$ is given as input to, but not controlled by, the collision finder. It is sufficient to focus on a single-layer network since any collision found in one layer propagates through subsequent layers.

Collision resistance in single-layer binary networks depends both on the activation function and the constraint density $\alpha = m/n$. Our results identify a scaling window for collision finding, governed by the interaction between the activation’s margin width $\kappa$ and the density $\alpha$. We analyze three representative activations under this lens:

\begin{itemize}
    \item The \emph{asymmetric binary perceptron} (ABP), $\varphi(z)=\sign(z)$.
    \item The \emph{symmetric binary perceptron} (SBP), $\varphi(z)=\sign(\kappa-|z|)$.
    \item A \emph{randomized oscillating activation}, where each neuron $i$ uses an odd number $K$ of independent random thresholds $t_{ij} \sim \mathrm{Unif}[-\kappa,\kappa]$, 
    \[
    \varphi_i(x)=\prod_{j=1}^{K}\sign(x-t_{ij})\,.
    \]
    Here $\kappa>0$ determines the interval of oscillation and the odd integer $K\ge1$ determines its complexity. The oddness makes the activation change sign between the two tails, rather than returning to the same value on both sides of the oscillating window. This is the regime we analyze: the OGP and online lower-bound proofs use odd \(K\) in the row-wise contraction argument, and we do not treat even \(K\). Each row $i$ of the output uses its own randomly chosen thresholds, so the activations $\varphi_i$ vary across coordinates. The thresholds are sampled once together with the instance and are then fixed and revealed; the activation is randomized only in the choice of the instance, not during evaluation.
\end{itemize}

When $\kappa \ll 1/\sqrt{\alpha}$, we design an online algorithm that finds \emph{extensive collisions} with a sufficiently small constant overlap gap, for any activation that is constant on the positive tail $[\kappa,\infty)$. This includes both ABP ($\kappa=0$) and SBP in the small-$\kappa$ regime (for SBP, sufficiently positive margins all lie on the same constant tail, so the outputs coincide).

For the randomized oscillating activation, in the complementary large-width regime, we show that the collision space undergoes an $r$-wise {\em overlap gap property} ($r$-OGP). The OGP is a geometric criterion linked to algorithmic barriers in random inference and optimization~\cite{GamarnikSudanAOP17,GJW20,GamarnikJagannathAOP21,SubagInvent17,Gamarnik21OGPsurvey}. In our setting, the relevant form is an \emph{ensemble} \(r\)-OGP, which applies to several correlated instances sharing a common prefix; this rules out online algorithms for finding extensive collisions in the corresponding parameter regime. In related models, such gaps also align with known limitations of low-degree polynomial methods and message-passing heuristics~\cite{GJW20,GamarnikJagannathAOP21}. Our formal lower bound here is for online algorithms; lower bounds for broader algorithmic models, including low-degree, message-passing, and quantum algorithms, remain open.


The SBP lies between these two regimes. For small \(\kappa\), it falls within our positive-tail algorithmic window; for very large \(\kappa\), the activation is nearly constant on typical Gaussian margins, and collisions become trivial. The difficult case is the intermediate window, where the positive-tail escape route no longer applies but no OGP-based hardness result is known for collision finding. Resolving whether SBP is collision resistant in this middle regime, for some \(\alpha<1\) and some \(\kappa\), remains a central open challenge.

\subsection{A Statistical Physics Perspective} 

The ABP is susceptible to a second preimage attack: sample $\vecy = \sign(\matA \vecx)$ for random $\vecx$ and try to solve
the so-called \emph{teacher-student problem} on input $(\matA,\vecy)$. The output will be an $\vecx'$ such that,
$\sign(\matA \vecx) = \vecy= \sign(\matA \vecx')$, which is a collision whenever $\vecx' \neq \vecx$. This occurs with probability at least $1/2$ for $m<n$ and with noticeable
probability up to $\alpha \leq 1.1$. Approximate message passing (AMP) algorithms are known to heuristically solve
the teacher-student problem in this regime~\cite{BILSZ15}. 
\begin{fact} The ABP is (heuristically) not collision resistant for any $\alpha < 1.1$.
\end{fact}

The perceptron and its random variants have long served as a model system in statistical physics \cite{Gardner88,GardnerDerrida88}; see also the monograph~\cite{MezardMontanariBook}. A recent line of work studies the \emph{square-wave perceptron} (SWP), where a periodic activation induces structured oscillations across the input domain~\cite{swp,collisionResistance2025}. Using non-rigorous replica methods alongside first-moment arguments, it predicts an overlap gap at large densities and suggests a phase transition in collision-finding complexity. 

Our work is complementary: we study non-periodic, \emph{randomized} oscillating activations with odd \(K\) and prove a rigorous $r$-OGP, while also giving a positive online algorithm for any activation that is constant on the positive tail $[\kappa,\infty)$ in the small-width regime and for sufficiently small overlap gap. In particular, the randomized (non-periodic) structure we consider avoids the lattice-like regularity of the SWP. This motivates the conjecture that a single-layer neural network with random weights, binary inputs, and a randomized oscillating activation is collision resistant in the parameter regime where the OGP first moment is negative. In this regime, our ensemble OGP result proves an online lower bound.

\subsection{Result I: An Online Collision Finder for Positive-Tail Activations}
For activations that are constant on a positive tail, collision finding is algorithmically easy at low densities. Indeed, if \(\varphi(x)\) is constant for all \(x\ge\kappa\), then it suffices to drive all margins above \(\kappa\); once there, the activation is flat and the outputs coincide. We design a simple online algorithm that efficiently produces two distinct inputs mapping to the same output pattern whenever \(\alpha\) is below a constant threshold.
 
 The algorithm operates in the \emph{extensive regime}, where the two colliding inputs differ on a linear fraction of their coordinates. This is also the regime relevant for our hardness results later: both our positive and negative results concern collisions that are extensive rather than so-called \emph{subextensive} collisions that differ in a vanishing fraction of coordinates.

\begin{theorem}[Online Collision Finding at Low Density]\label{thm:informal-online}
There exists \(\delta_0>0\) such that for every fixed \(\delta\in(0,\delta_0)\) there are constants \(\alpha_0=\alpha_0(\delta)>0\) and \(c>0\) such that, for every \(\alpha<\alpha_0\) and every margin parameter \(\kappa\le c/\sqrt{\alpha}\), there is an efficient online algorithm which, given a random matrix \(\matA\in\mathbb{R}^{m\times n}\) with \(m=\alpha n\), outputs (with constant probability) two inputs \(\vecx,\vecy\in\{\pm1\}^n\) satisfying:
\begin{align*}
    \left\lvert\vecx^\top\vecy\right\rvert\le (1-\delta)\,n\,,
    \quad\quad
    \min_{i}(\matA \vecx)_i \ge \kappa\,, \quad\quad \min_{i}(\matA \vecy)_i \ge \kappa\,.
\end{align*}
In particular, for any activation function $\varphi$ that is constant on $[\kappa,\infty)$, we have $\varphi(\matA\vecx)=\varphi(\matA\vecy)$, so $\vecx$ and $\vecy$ form a $\delta$-extensive collision for $\varphi$.
\end{theorem}

Our algorithm works for both the ABP and the SBP for sufficiently small densities. For ABP, the same procedure, when run on a single vector, also finds a preimage of the all-positive output $\mathbf{1}^m$.

The algorithm's evolution is governed by an exponential potential that gives a single global measure of how many coordinates still have poor margin. In the collision-finding variant, we first create the desired overlap gap by assigning opposite signs to $\vecx$ and $\vecy$ on an initial block of coordinates. From that point on, every new coordinate is assigned the same sign in both vectors, so the overlap remains fixed while the margins continue to improve. The analysis then follows the pointwise minimum of the two margin vectors: a coordinate is considered safe only once both $\matA\vecx$ and $\matA\vecy$ are safely above the target margin. The \(Push\) phase steadily reduces the exponential potential of this minimum state, shrinking the low-margin tail. A final refinement phase focuses on the remaining bad coordinates and raises them above the threshold, ensuring that every entry of both $\matA\vecx$ and $\matA\vecy$ is at least $\kappa$.

This potential-based analysis gives a different way to study sequential margin-building algorithms. An earlier algorithm of Kim and Roche~\cite{KimR95} also constructs signs in stages: at each stage it examines a fresh block of columns and chooses the corresponding signs by majority vote among the rows with the smallest current partial margins. Their proof tracks the evolution of this low-margin tail through the stages and shows that all margins become positive when \(\alpha=O(1/\log^2 n)\), achieving minimum margin \(\Omega(1/\log n)\). Our exponential-potential argument serves a similar purpose, but reduces the low-margin tail to a single quantity that is easier to adapt to the collision-finding setting.

A recent algorithm of Schramm et al.~\cite{10.1145/3717823.3718124} gives a sharper result for the single-preimage margin problem using tools from discrepancy minimization. Their method first solves a linear program over the relaxed cube $[-1,1]^n$, with a random objective direction and the desired margin constraints, and then rounds the remaining fractional coordinates while preserving feasibility. For every constant $c<2/\pi$ and all sufficiently large $\kappa$, it finds $\vecx$ with $\matA\vecx\ge \kappa\mathbf{1}$ whenever $\alpha\le c/\kappa^2$, matching the large-margin statistical threshold up to the sharp constant. Our potential-based method has a less optimized constant for this single-preimage task, but it is tailored to collision finding: by tracking the pointwise minimum of two evolving margin vectors, it constructs two correlated inputs with prescribed extensive overlap, rather than a single preimage of $\mathbf{1}^m$. Whether the sharper discrepancy-minimization approach can be adapted to this collision-finding setting remains open.

\subsection{The Symmetric Binary Perceptron}

The SBP exposes the main limitation of the positive-tail framework. Our online algorithm applies whenever $\kappa \ll 1/\sqrt{\alpha}$, yielding efficient collisions with a sufficiently small constant overlap gap in the small-$\kappa$ regime. Separately, in extreme parameter regimes where the output can be saturated to $\pm\mathbf{1}^m$, earlier margin-based algorithms~\cite{BS20,KimR95,10.1145/3717823.3718124} also produce collisions. Beyond these extremes, however, the collision resistance of SBP remains unresolved. Neither known algorithmic techniques nor hardness arguments currently apply when $\kappa$ lies in an intermediate window, such as $\kappa \approx 0.675$ and $\alpha \approx 0.9$, where statistical physics suggests nontrivial structure may emerge.

One insight into the structure of SBP is provided by Gamarnik et al.~\cite{GKPX22}, who show that in the teacher–student version of SBP, the solution space satisfies an $m$-overlap gap property (m-OGP) at very small $\kappa$ and $\alpha$. This implies that, in this setting, \emph{second preimage attacks}, where one first samples $\vecx$, computes $\varphi(A\vecx)$, and seeks a new $\vecx'$ mapping to the same output, cannot succeed efficiently. However, this does not rule out collision finding in the standard formulation, where the collision-finder has full control over the two inputs and is not restricted to second-preimage attacks. Nor does their result extend to moderate or large values of $\alpha$.

A related hardness indication comes from recent work of Vafa and Vaikuntanathan~\cite{VV25}, who give a worst-case-to-average-case reduction from standard approximate lattice problems to the SBP search problem. Translating their notation to ours, their polynomial-time hardness result applies to margins of order \(\kappa=\alpha^{1/2+\epsilon}\) for every constant \(\epsilon>0\), under polynomial-factor worst-case lattice hardness assumptions; under stronger subexponential lattice hardness assumptions, they obtain hardness nearly at the Bansal--Spencer scale \(\kappa\approx\sqrt{\alpha}\), up to logarithmic factors. Their reduction applies in a different regime, where the number of variables is polynomially larger than the number of constraints, and they explicitly leave the proportional regime \(m=\Theta(n)\) open. Thus their result does not directly address collision finding or the intermediate-density SBP window considered here, but it provides complementary evidence for SBP hardness from worst-case lattice assumptions.

\begin{quote}
\noindent{\bf Open Question.}
{\em For some $\alpha<1$, does there exist an intermediate window of $\kappa$ with
$1/\log n \ll \kappa \ll 1/\sqrt{\alpha}$ in which SBP is collision resistant
against polynomial-time algorithms?}
\end{quote}

We therefore introduce a randomized activation with the same kind of non-monotone, windowed behavior as SBP, but with a collision space that can be analyzed by a first-moment argument. The randomized oscillating activation replaces the fixed SBP boundary by independent random thresholds at each neuron. This extra randomness removes the rigid alignments that make the SBP collision space difficult to analyze, while preserving the obstruction to the positive-tail algorithm. In this model we prove an overlap gap property for extensive collisions, and hence an online lower bound, in a regime where our algorithmic approach no longer applies.

\subsection{Result II: Overlap Gaps for Randomized Oscillating Activation}

To obtain a model in which the collision geometry can be analyzed by a first-moment argument, we turn to the randomized oscillating activation introduced above. The oscillations persist across a wide region of inputs, and when $\kappa$ is sufficiently large the positive-tail algorithm cannot drive all margins above $\kappa$ within the available $n$ columns. The row-wise random thresholds make the activation non-monotone and heterogeneous while keeping the one-row collision probability tractable.

We prove that for this randomized model, extensive collision pairs exhibit an \emph{overlap gap property} (OGP). The theorem considers collision pairs that are internally well separated, and shows that the mutual overlap between two such pairs cannot lie in an intermediate range. In this sense, the collision space has a forbidden band of external overlaps: valid pairs can be close to one another or far apart, but not at intermediate separation. Such overlap gaps are known to create algorithmic barriers in spin-glass models and random constraint satisfaction problems.

\begin{theorem}[Overlap Gap Property]\label{thm:informal-ogp}
For every constant density \(\alpha\in(0,1)\), for suitable constants \(r\ge2\), \(\delta>0\), and \(C>0\) depending on \(\alpha\), whenever \(K\) is odd, \(n\) is sufficiently large, and the parameters satisfy
\[
\kappa^2\ge \frac{C}{\alpha}
\qquad\text{and}\qquad
K\ge C\,\frac{\kappa}{\sqrt{\delta}},
\]
the randomized oscillating activation exhibits an $r$-wise overlap gap property for $\delta$-extensive collisions with high probability over $\matA$ and the random thresholds.
\end{theorem}

The formal statement in \Cref{thm:randomized-ogp-first-moment} gives the sharper first-moment condition, expressed through the row exponent \(\Lambda_{r,\delta}(\kappa,K)\). The constant \(C(\alpha)\) above comes from choosing \(r=O(1/\alpha)\) and from the constants needed for this exponent to dominate the entropy term; in particular, it absorbs the factor \(2^r r^3\) appearing in the sufficient lower bound on \(K\) in \Cref{corollary:rogp}.

\begin{theorem}[Online Collision Resistance]\label{thm:informal-hard}
For every constant density \(\alpha\in(0,1)\), for suitable constants \(\delta>0\) and \(C>0\) depending on \(\alpha\), whenever \(K\) is odd, \(n\) is sufficiently large, and the parameters satisfy
\[
\kappa^2\ge \frac{C}{\alpha}
\qquad\text{and}\qquad
K\ge C\,\frac{\kappa}{\sqrt{\delta}},
\]
every online algorithm has success probability at most $2^{-\Omega(\alpha n)}$ of outputting a $\delta$-extensive collision on a single random instance. In particular, the randomized family $\vecx\mapsto\varphi(\matA\vecx)$ is online $\delta$-extensive collision resistant at density $\alpha$.
\end{theorem}

The constant \(C\) above comes from the same first-moment and ensemble-OGP requirements as in the preceding theorem; the precise dependence is given by the row exponent \(\Lambda_{r,\delta}(\kappa,K)\) in \Cref{thm:randomized-ogp-first-moment}. The formal online bound is \(2^{-\Omega(n/r)}\); for fixed \(\alpha\) the proof takes \(r=O(1/\alpha)\), which gives the displayed \(2^{-\Omega(\alpha n)}\).

Our formal lower bound is stated for extensive collisions. For cryptographic collision resistance in the usual sense, one may first encode the input with a code \(E:\{0,1\}^{\ell}\to\{\pm1\}^n\) whose distinct codewords have absolute inner product at most \((1-\delta)n\). Any collision in the composed function \(\varphi\circ\matA\circ E\) then gives a \(\delta\)-extensive collision in the neural network. If the code has rate \(R=\ell/n\), the composed function is compressing precisely when \(\alpha<R\); hence this interpretation applies in the density regimes covered by our lower bound whenever such a code exists with rate exceeding \(\alpha\).

The randomized oscillating activation differs from the square-wave perceptron in that the thresholds are sampled independently for each row, rather than placed on a fixed period shared across all neurons.

The online lower bound provides evidence for a broader hardness phenomenon. The square-wave model studied in~\cite{swp} suggested that OGP can serve as a geometric indicator of collision resistance in neural networks, supported there by first-moment calculations and AMP experiments. Our randomized oscillating model removes the shared periodic structure of the square wave and proves the corresponding online statement: once the row exponent dominates entropy, the collision space develops an overlap gap, and online algorithms fail. Our theorem proves this online implication for odd \(K\); the conjectural next step is that the same geometric obstruction reflects hardness beyond the online model. In particular, even for fixed \(\kappa\) and fixed odd \(K\), sufficiently high density should make it computationally hard to coordinate two extensive preimages through the row-wise random thresholds. We do not prove lower bounds against spectral, algebraic, lattice-based, low-degree, message-passing, or quantum algorithms, but the OGP theorem motivates the following conjecture.

\begin{conjecture}[Hardness of Collision Finding]\label{conj:constantK}
For fixed constants $K \in \mathbb{N}$, $\kappa > 0$, and $\delta > 0$, there exists a threshold $\alpha_\star < 1$ such that for all $\alpha \ge \alpha_\star$, no polynomial-time algorithm can find a $\delta$-extensive collision with non-negligible probability.
\end{conjecture}

We leave confirming this conjecture, or finding algorithmic counterexamples, as an open direction for future work.

\paragraph{Paper Organization.} The preliminaries fix notation and the formal collision and OGP definitions. Section~\ref{sec:evolution} proves the positive result via a one-sided potential method and a refinement procedure for positive-tail activations. Section~\ref{sec:random-activation} proves the first-moment OGP bound for the randomized oscillating activation; the subsequent subsections extend it to correlated-prefix ensembles and derive the online lower bound.

\section{Technical Overview}

The two main technical components of this paper are a positive algorithmic result for positive-tail activations, including the asymmetric binary perceptron (ABP), and a negative online-hardness result via the overlap gap property (OGP) for a randomized activation function. We outline the core ideas behind both directions.

\subsection{A Simple Online Algorithm}
For any activation that is constant on the positive tail $[\kappa,\infty)$, we only need to drive all margin values $(\matA\vecx)_i$ and $(\matA\vecy)_i$ above $\kappa$. Once there, the activation $\varphi$ stabilizes and outputs coincide. The drift analysis is stated in the normalization $\matB=\matA/\sqrt{\alpha}$, whose columns have entrywise variance $1/m$; reaching margin $\rho=\kappa/\sqrt{\alpha}$ for $\matB$ gives margin $\kappa$ for $\matA$. The basic margin-building routine processes columns online, maintaining the current normalized margin vector \(\vecs=\matB_{\le t}\vecx_{\le t}\). At step \(t+1\), it chooses the next bit \(x_{t+1}\) so as to reduce an exponential potential \(\Phi(\vecs)\), which gives large weight to coordinates whose margins are still small. This local sign choice creates a systematic drift in the potential: although each normalized column changes every margin by only \(O(1/\sqrt m)\), the exponential weighting turns these small updates into a \(\Theta(1/m)\) decrease in \(\Phi\) per step in expectation.

We show that when the activation $\varphi$ is constant on $[\kappa,\infty)$, one can efficiently construct extensive collisions with a sufficiently small constant overlap gap in the shrinking regime $\alpha = m/n < 1$. The algorithm processes the columns of $\matA$ sequentially, choosing the signs of two vectors $\vecx,\vecy\in\{\pm1\}^n$ in an online fashion as each column is revealed. The goal is to produce two inputs that differ on a linear fraction of their coordinates, yet agree under the activation map $\varphi(\matA\cdot)$.

For a single margin vector, the one-step drift is captured by the following potential. Given the current state \(\vecs_t\), the algorithm chooses the next sign \(x_t\in\{\pm1\}\) to minimize
\[
\Phi(\vecs) = \frac{1}{m}\sum_{i=1}^m e^{-\lambda s_i}\,,
\]
which penalizes small margins and rewards those that already lie far above $\rho$. The exponential weights make the potential sensitive to precisely the coordinates that are still problematic: a coordinate with small margin contributes much more to \(\Phi\) than a coordinate already far above the target. When a fresh Gaussian column arrives, the algorithm chooses the sign that is better aligned with this weighted set of low-margin coordinates. Concretely, after seeing the fresh column \(\veca\), the choice \(x_t=\sign\!\left(\frac1m\sum_i a_i e^{-\lambda s_i}\right)\) makes the first-order term in the Taylor expansion negative and of order \(\Theta(1/m)\) in expectation over \(\veca\). The second-order correction is also of order \(1/m\), but with a smaller constant at our choice of \(\lambda\), and the remaining terms are \(O(1/m^2)\). This yields a clean multiplicative drift for \(\E[\Phi]\). Quantitatively,
\[
    \mathbb{E}\!\left[\Phi(\vecs_{t+1}) \,\middle|\, \vecs_t\right]
    \le \left(1 - \frac{\gamma}{m}\right)\Phi(\vecs_t),
    \quad
    \gamma=\Theta(1)\,,
\]
because $\sqrt{2/\pi}\,\sqrt{\E\left[e^{-2\lambda s}\right]}-\tfrac12\lambda\,\E\left[e^{-\lambda s}\right]\ge c\,\E\left[e^{-\lambda s}\right]$ for $\lambda=\sqrt{2/\pi}$ and some universal $c>0$. The first term is the gain from choosing the better sign against a Gaussian column, while the second term is the variance cost from the quadratic Taylor expansion. Thus, $\Phi$ decays exponentially in $t/m$. After \(O((1+\rho)m)\) steps, only a small fraction of coordinates of $\matB\vecx$ remain below the normalized target margin $\rho$; a short refinement phase completes the separation, and rescaling by $\sqrt{\alpha}$ gives margins above $\kappa$ for $\matA$.

For collision finding, we first create the required overlap gap by assigning opposite signs to $\vecx$ and $\vecy$ on $\beta_1m$ coordinates, with $\beta_1=\delta/(2\alpha)$. Since $m=\alpha n$, this gives $\delta n/2+o(n)$ initial disagreements. After this initialization, every remaining column receives the same sign in $\vecx$ and $\vecy$, so the overlap is preserved. The proof tracks the pointwise minimum $\vecs=\min(\matB\vecx,\matB\vecy)$; because $\min(u+za,v+za)=\min(u,v)+za$, the same one-sided potential argument applies to this minimum state. A refinement phase based on \(FocusedPush\) uses $O(m)$ additional columns to clear the remaining low-margin coordinates. The total number of used columns is at most $n$ once $\delta$, $\kappa\sqrt{\alpha}$, and $\alpha$ are chosen sufficiently small: reaching normalized margin \(\rho=\kappa/\sqrt{\alpha}\) costs \(O((1+\rho)m)=O((\alpha+\kappa\sqrt{\alpha})n)\) drift columns, and the initialization/refinement budgets are also linear in \(m\) and \(\delta n\). This is why the theorem is stated only for sufficiently small constant overlap gap. Once all entries of $\matA\vecx$ and $\matA\vecy$ lie above $\kappa$, the component-wise activations stabilize, giving,
\[
\varphi(\matA\vecx) = \varphi(\matA\vecy)\,,
\]
certifying a $\delta$-extensive collision. Informally, the exponential potential converts the low-margin tail into a single scalar quantity: coordinates far below the target contribute exponentially more than coordinates already safely above it. The \(1/m\) drift therefore does not just improve an average margin; it steadily reduces the weighted mass of problematic coordinates. After \(O((1+\rho)m)\) drift columns this leaves only a small residual set, which the refinement phase targets directly.

\subsection{Overlap Gaps for Collision Finding}

To certify online hardness for more expressive, non-monotone activations, we establish an $r$-wise overlap gap property (OGP) for a randomized multi-threshold activation. For each row $i\in[m]$, we fix an odd integer $K\ge1$, so that the activation changes sign between the two tails, draw independent thresholds $t_{ij}\sim\mathrm{Unif}[-\kappa,\kappa]$ for $j\in[K]$, and define
\[
\varphi_i(x)=\prod_{j=1}^{K}\sign(x-t_{ij})\,,\qquad
\varphi(\matA\vecx)_i=\varphi_i((\matA\vecx)_i)\,.
\]
The independent thresholds make the row calculation depend on the Gaussian positions of the candidate margins, rather than on a common period shared by all rows. For a fixed admissible \(r\)-tuple, a single threshold has some chance of separating the signs in the selected Gaussian subfamily. Using \(K\) independent thresholds amplifies the resulting row-wise contraction, while increasing \(\kappa\) spreads the thresholds over a wider interval and weakens the contribution of each one. The exponent \(\Lambda_{r,\delta}(\kappa,K)\) is the resulting row-wise gain. The OGP proof then compares this gain with the entropy cost of choosing an admissible \(r\)-tuple.

Intuitively, the OGP states that no valid $r$-tuple of collision pairs $(\vecx_i,\vecy_i)$ can have external overlaps,
\[
\frac{1}{2n}\left(\vecx_i^\top\vecx_j+\vecy_i^\top\vecy_j\right)\,,
\]
lying in a fixed forbidden interval. The proof certifies this gap by a first-moment bound. Fix $r$ collision pairs obeying the internal overlap constraint and consider a random row $\veca^\top$ of $\matA$. For each pair,
\[
(X_i,Y_i)=(\veca^\top\vecx_i,\veca^\top\vecy_i)\,,
\]
is jointly Gaussian with covariance determined by overlaps. For a single pair, the signs of \(X_i\) and \(Y_i\) differ exactly when the threshold lies between them. Conditional on \(|X_i|,|Y_i|\le \kappa\), this has probability \(\Theta(|X_i-Y_i|/\kappa)\). The proof uses a higher-order version of this observation: once a selected subfamily of Gaussian values is well separated inside \([-\kappa,\kappa]\), a random threshold has nontrivial probability of making the relevant product of signs negative. The \(K\) independent thresholds then amplify this one-threshold contraction.
By \Cref{lem:gaussian-estimates}, for every $p\in(0,1)$ the selected Gaussian variables are separated by order $p\sqrt{\delta}$ except with probability of order $r^2p$, while the event that some selected Gaussian lies outside $[-\kappa,\kappa]$ costs order $r\,e^{-\kappa^2/2}$. Combining this with the $K$ independent thresholds gives the formal row-wise exponent $\Lambda_{r,\delta}(\kappa,K)$ defined in Section~\ref{sec:random-activation}: for a fixed admissible $r$-tuple, the one-row collision probability is at most \(2^{-\Lambda_{r,\delta}(\kappa,K)}\).
By the entropy count in \Cref{lem:admissible-entropy-count}, summing over all families of $r$ pairs consistent with the overlap constraints gives the per-coordinate first-moment bound,
\[
\frac1n\log_2
\mathbb{E}[\#\text{$r$-tuples of $\delta$-admissible collisions}]
\le
2+2(r-1)H(\delta/(2c_1))
-\alpha\Lambda_{r,\delta}(\kappa,K)
+o(1).
\]
Thus the first-moment exponent is negative once $\alpha\Lambda_{r,\delta}(\kappa,K)$ exceeds $2+2(r-1)H(\delta/(2c_1))$ by a constant amount. By Markov's inequality, in this negative-exponent regime no such $r$-tuple exists with exponentially high probability, giving the formal OGP result in \Cref{thm:randomized-ogp-first-moment,corollary:rogp}. The overlap assumptions enter through a selection lemma. Applied to the \(2r\) Gaussian projections from one row, it shows that admissibility forces an \((r+1)\)-variable subfamily whose covariances are bounded away from one. This separated subfamily is what drives the one-row contraction; the anti-concentration parameter \(p\) converts covariance separation into actual separation of Gaussian values and is optimized in the definition of \(\Lambda_{r,\delta}(\kappa,K)\).

Online hardness uses a separate ensemble version of the OGP. In the ensemble setting, we consider correlated instances $(\matA^{(1)},\ldots,\matA^{(r)})$ sharing their first $t$ columns and having independent suffixes, with
\[
\frac{\delta}{c_2}<1-\frac{t}{n}<\frac{\delta}{c_1}.
\]
For prefix-aligned families, the first-moment count becomes \(2+2(r-1)(1-t/n)+o(1)\), and the same row-wise contraction applies after selecting an explicitly separated subfamily. This gives ensemble \(r\)-OGP. If a single-instance online algorithm succeeds with probability $p$, then, after conditioning on the shared prefix, thresholds, and internal coins, the $r$ suffixes remain independent with a common conditional success probability \(q\). Jensen's inequality gives \(\E[q^r]\ge \E[q]^r=p^r\), so the coupled experiment succeeds on all \(r\) instances with probability at least \(p^r\). Because the algorithm is online and the coupled instances have the same prefix, the outputs produced on the first \(t\) coordinates are identical across the \(r\) runs. On the suffix, we independently flip columns in each coupled instance and then undo the flips in the output; this preserves the collision equations, while making the suffix contributions to external overlaps behave like sums of independent mean-zero signs. Choosing $t$ inside the forbidden window therefore turns simultaneous success into an ensemble-OGP-forbidden configuration except with exponentially small probability. Hence \(p\le \exp(-\Omega(n/r))\).

In contrast to the ABP, where monotonicity and a simple exponential potential enable efficient collision finding, sufficiently strong randomized oscillation introduces a geometric barrier for online algorithms.

\section{Preliminaries and Definitions}

We will use capital letters $A,B,C, \ldots$ to denote sets and use standard notation $\cap, \cup, \setminus$ to denote common set-theoretic operations. We denote vectors using boldface lower-case letters $\vecx,\vecy, \ldots, \in \bbR^n$, and reserve boldface upper-case letters $\matA, \matB, \ldots$ for matrices. For a vector \(\vecx\), write \(\vecx_{\le t}\) for its first \(t\) coordinates; for a matrix \(\matB\), write \(\matB_{\le t}\) for its first \(t\) columns. We use $\boldsymbol{0}_k$ or $\boldsymbol{1}_k$ to denote the $k$-vectors containing only 0 or 1 respectively. Inner products between vectors $\vecx$ and $\vecy$ are written as $\vecx^\top \vecy$, or occasionally as $\langle \vecx, \vecy \rangle$ when the latter is more convenient. We may omit the subscript if its length is implicit from the context. We use $\boldsymbol{1}\{\cdot\}$ to denote the indicator variable which is 1 if its argument holds true. For $\vecx,\vecy\in\{-1,1\}^n$ we write $\Delta_{\mathrm{Ham}}(\vecx,\vecy)$ for their Hamming distance; equivalently, $\vecx^\top\vecy = n-2\Delta_{\mathrm{Ham}}(\vecx,\vecy)$. We write $[n]=\{1,\ldots,n\}$ for positive integers $n$.

We denote by $\calN(\mu, \sigma^2)$ the Gaussian distribution with mean $\mu \in \bbR$ and variance $\sigma^2 > 0$. Similarly, we denote by $\calN(\boldsymbol{\mu}, \bSigma)$ the multivariate Gaussian distribution with mean $\boldsymbol{\mu} \in \bbR^{n}$ and covariance matrix $\bSigma \in \bbR^{n \times n}$. Denote by $H(p) = -p\log_2(p) - (1-p) \log_2(1-p)$ the binary entropy function. Throughout the paper, $\alpha=m/n$ is treated as a fixed constant unless stated otherwise. When a fixed density $\alpha$ is specified, all asymptotic statements are over integer pairs $(m,n)$ with $n\to\infty$ and $m/n\to\alpha$; equivalently, one may take $m=\lfloor \alpha n\rfloor$, with rounding effects absorbed in the $o(1)$ terms. We use the convention \(\sign(0)=1\).

\paragraph{Collision Finding.}
Fix an activation map $\varphi : \bbR^m \to \{-1,1\}^m$, deterministic coordinate-wise or row-wise randomized. The instance consists of a matrix $\matA \in \bbR^{m \times n}$ with entries $\matA_{ij} \sim \calN(0,1/n)$, drawn independently. In the randomized case, the activation is sampled once and fully revealed together with $\matA$. When randomized, all thresholds are drawn once at the start and then treated as fixed.

\begin{definition}[Collision Problem]
    Given $\delta \in (0,1)$, a (randomized) algorithm $C$ \emph{solves the collision problem} for $\varphi$ at distance $\delta$ if, on input $\matA$ and any sampled activation thresholds, it outputs $\vecx,\vecy \in \{\pm1\}^n$ with
    \begin{equation} \label{eq:ogp-def}
        \varphi(\matA \vecx) = \varphi(\matA \vecy)\,,
    \quad\text{and}\quad
    \left|\vecx^\top \vecy\right| \le (1-\delta)\,n\,,
    \end{equation}
    with non-negligible probability over the joint randomness of $\matA$ and $C$. We say $\varphi$ is \emph{$\delta$-extensive collision resistant} if no probabilistic polynomial-time algorithm solves this distance-$\delta$ collision problem.
\end{definition}

The preceding definition is the full probabilistic-polynomial-time notion for extensive collisions, stated for context. The rigorous lower bound proved in this paper is for online algorithms; when referring to proved hardness below, we use \emph{online extensive collision resistance} to mean resistance within that restricted online model.

We refer to any pair $(\vecx,\vecy)$ satisfying \cref{eq:ogp-def} as a \emph{$\delta$-extensive collision}: a valid collision whose two inputs have absolute inner product at most $(1-\delta)n$. When $\delta=\Omega(1)$ is a positive constant we call it simply an \emph{extensive collision}.

The extensive-distance requirement is weaker than standard cryptographic collision resistance, since it only rules out collisions whose inputs have small absolute overlap. This distinction can be removed by placing an error-correcting code before the neural network, provided the code itself has the corresponding absolute-overlap separation and the resulting function is still compressing. Indeed, if \(E:\{0,1\}^{\ell}\to\{\pm1\}^n\) satisfies
\[
\left|\langle E(u),E(v)\rangle\right|\le (1-\delta)n
\qquad\text{for all }u\neq v,
\]
then any collision
\[
\varphi(\matA E(u))=\varphi(\matA E(v)),
\qquad u\neq v,
\]
induces a \(\delta\)-extensive collision between the codewords \(E(u)\) and \(E(v)\). If the code has rate \(R=\ell/n\), the composed function \(\varphi\circ \matA\circ E\) maps \(\ell\) bits to \(m=\alpha n\) bits and is compressing exactly when \(\alpha<R\). Since such absolute-overlap codes have rates bounded away from zero for every fixed small enough \(\delta\), the ECC interpretation applies to the randomized-activation online lower bound for any fixed \(\alpha\) below such a rate, by taking \(\kappa=\Omega(1/\sqrt{\alpha})\) and an odd \(K=\Omega_\alpha(\kappa/\sqrt{\delta})\).

\begin{definition}[Overlap Gap Property]\label{def:ogp}
    Let $r\in\bbN$, $0< \delta < 1$, and let $\varphi:\bbR^m\to\{-1,1\}^m$ be an activation map, deterministic coordinate-wise or sampled once and then fixed. Let $2<c_1<c_2<2c_1$ be two constants. Define $T(r, \delta, c_1, c_2) \subseteq \{-1,1\}^{2rn}$ as the set of all $r$ pairs $(\vecx_i, \vecy_i)_{i=1}^r$ with $\vecx_i, \vecy_i \in \{-1,1\}^n$ that satisfy the following two conditions.
    \begin{enumerate}
        \item \emph{(Internal Overlap)}. For every $i=1\ldots r$, it holds that,
        \begin{equation}\label{eq:extensive_distance}
            \left \lvert \vecx_i^\top \vecy_i \right \rvert  \leq (1-\delta)\,n\,.
        \end{equation}
        \item \emph{(External Overlap)}. For every $i,j =1\ldots r$ with $i\neq j$, we have that, 
        \begin{equation}\label{eq:pairwise_overlap}
        1 - \frac\delta{c_1} \leq \frac{\vecx_i^\top\vecx_j + \vecy_i^\top\vecy_j}{2n} \leq 1 - \frac\delta{c_2}\,.
        \end{equation}
    \end{enumerate}
    A set of $r$ pairs $(\vecx_i, \vecy_i)_{i=1}^r$ satisfying the internal and external overlap requirements is said to be \emph{$\delta$-admissible}. 
    
    Let $\alpha > 0$ be a density parameter, and let $m=m(n)$ satisfy $m/n\to\alpha$. Let $\matA = (\veca_j)_{j=1}^n \in \bbR^{m \times n}$ have columns $\veca_j \in \bbR^m$, and entries $\matA_{ij}\, {\scriptstyle \iid}\, \calN(0,1 / n)$. Denote by $S(r, \delta, c_1, c_2; \matA) \subseteq T(r, \delta, c_1, c_2)$ the subset of $\delta$-admissible sets of $r$ pairs that also satisfy a third property.
    \begin{enumerate}
        \item[3.] \emph{(Collision)}. For every $i=1\ldots r$, it holds that,
        \begin{equation}
            \varphi(\matA \vecx_i) = \varphi(\matA \vecy_i)\,.
        \end{equation}
    \end{enumerate}
    For fixed \(\matA\) and fixed activation thresholds, the function \(f_\matA(\vecx)=\varphi(\matA\vecx)\) satisfies the \emph{\(r\)-overlap gap property} (\(r\)-OGP) with parameters \(0<\delta<1\) and \(2<c_1<c_2<2c_1\) if
    \[
    S(r,\delta,c_1,c_2;\matA)=\emptyset,
    \]
    i.e. if there is no \(r\)-tuple of collision pairs satisfying the internal overlap condition whose pairwise external overlaps all lie in the forbidden window.
\end{definition}

\paragraph{Probability Conventions.}
Unless stated otherwise, all probabilities and expectations are over the joint randomness of:
(i) the Gaussian matrix $\matA$ (and, in the ensemble setting, the correlated draw $(\matA^{(1)},\dots,\matA^{(r)})$),
(ii) any random thresholds defining the activation (sampled once and then revealed together with $\matA$),
and (iii) the internal randomness of the algorithm (when present).
When we write $\Pr[\cdot \mid \matA]$ (or $\mathbb{E}[\cdot \mid \matA]$) we condition on the realized instance(s)
and thresholds, and keep only algorithmic randomness.
When we write $\mathbb{E}[\mathsf{Coll}_{r,\delta}]$ we mean expectation over the draw of $\matA$ and the
activation thresholds (and over the ensemble when applicable).

\begin{lemma}[First-Moment Bound]\label{lem:first-moment}
If $X\ge0$, then $\Pr[X\ge a]\le \E\left[X\right]/a$ for every $a>0$. In particular, if $N$ counts forbidden configurations, then $\Pr[N\ge1]\le \E\left[N\right]$.
\end{lemma}

\begin{lemma}[Gaussian Estimates]\label{lem:gaussian-estimates}
If $G\sim\calN(0,\sigma^2)$, then for all $u\ge0$,
\[
\Pr[|G|\ge u\sigma]\le 2\,e^{-u^2/2},
\qquad
\Pr[|G|\le u\sigma]\le \sqrt{\frac{2}{\pi}}\,u.
\]
Moreover, for every fixed integer $k\ge1$, $\E\left[|G|^k\right]=\Theta_k(\sigma^k)$.
\end{lemma}

\begin{lemma}[Hoeffding Bound]\label{lem:hoeffding}
If $\xi_1,\ldots,\xi_s$ are independent mean-zero random variables with $\xi_i\in[-1,1]$, then for all $\lambda>0$,
\[
\Pr\!\left[\left|\sum_{i=1}^s \xi_i\right|\ge \lambda s\right]
\le 2\,e^{-\lambda^2s/2}.
\]
\end{lemma}

These standard forms can be found, for example, in \cite{Vershynin2018,DubhashiPanconesi2009}.

\section{Finding Collisions for Positive-Tail Activations}
\label{sec:evolution}
We demonstrate that extensive collisions with sufficiently small constant overlap gap can be found at sufficiently low aspect ratios $\alpha$ for activations that are constant on a positive tail.  To measure quality, we define the \emph{algorithmic rate} to be the smaller of the absolute-overlap gap $\delta$, where $|\langle x,x'\rangle|\le (1-\delta)n$, and the collision ratio $\alpha = m/n$ as $n$ increases.

We describe two algorithms. The first one works well in practice, but possibly not so well in theory.  We prove its correctness at inverse logarithmic algorithmic rates.  We leave open its performance at constant algorithmic rates.  We use it as a component in a more complicated algorithm that provably works at fixed positive algorithmic rates, for sufficiently small constants $\delta$ and $\alpha$. The main result of this section is the following:

\begin{theorem}[Collision Finding for Positive-Tail Activations]\label{thm:algorithm}
There exists $\delta_0>0$ such that for every fixed $\delta\in(0,\delta_0)$ there exist constants $\alpha_0=\alpha_0(\delta)>0$ and $c>0$, and an efficient online algorithm, such that for every pair of parameters $(\alpha,\kappa)$ satisfying $\alpha<\alpha_0$ and
\(
    \kappa \;\le\; c/\sqrt{\alpha}
\)
and for all sufficiently large $n$, the algorithm outputs with constant probability two vectors $\vecx,\vecy\in\{\pm1\}^n$ such that
\begin{align*}
    \left\lvert\vecx^\top\vecy\right\rvert\le (1-\delta)\,n\,,
    \quad\quad
    \min_{i}(\matA \vecx)_i \ge \kappa\,, \quad\quad \min_{i}(\matA \vecy)_i \ge \kappa\,.
\end{align*}
In particular, for any activation function $\varphi$ that is constant on $[\kappa,\infty)$, $\vecx$ and $\vecy$ form a $\delta$-extensive collision for $\varphi$.
\end{theorem}

The proof works in the normalized matrix $\matB=\matA/\sqrt{\alpha}$ and targets margin $\rho=\kappa/\sqrt{\alpha}$. The collision finder first creates the desired overlap gap by assigning opposite signs on an initial block, and then assigns the same signs to $\vecx$ and $\vecy$ on all remaining coordinates. The state tracked after initialization is the pointwise minimum $\min(\matB\vecx,\matB\vecy)$. A drift phase applies \(Push\) to reduce the one-sided exponential potential of this minimum state, and \(FocusedPush\) clears the remaining low-margin coordinates. The column budget is kept below $n$ by taking $\delta$, $\kappa\sqrt{\alpha}$, and $\alpha$ sufficiently small.

Because the drift analysis acts row-by-row, the algorithmic bookkeeping is measured relative to the output dimension $m$. Since $m=\alpha n$, a disagreement on a $\beta m$-sized set of coordinates corresponds to an overlap deficit of $2\alpha\beta+o(1)$ on the $n$-scale.

The drift analysis is most naturally stated in the normalization where the online columns have law $\mathcal{N}(\mathbf{0},\matI_m/m)$. To apply it to the paper's convention $\matA_{ij}\sim\calN(0,1/n)$, write $\matB=\matA/\sqrt{\alpha}$, so that $\matB_{ij}\sim\calN(0,1/m)$. If the algorithm constructs $\vecx,\vecy$ with $\min_i(\matB\vecx)_i,\min_i(\matB\vecy)_i\ge\rho$, then the original matrix satisfies $\min_i(\matA\vecx)_i,\min_i(\matA\vecy)_i\ge\sqrt{\alpha}\rho$. Thus, to obtain margin $\kappa$ for $\matA$, the $\matB$-normalized analysis targets $\rho=\kappa/\sqrt{\alpha}$.
Since the normalized process has \(n=m/\alpha\) available columns, a target \(\rho\) costs \(O(\rho m)\) drift steps, which is feasible as long as \(\alpha\rho=\kappa\sqrt{\alpha}\) is sufficiently small.

\subsection{A Simple Online Algorithm}
At the heart is an online search algorithm in the normalized model. Let $\matM\in\bbR^{m\times n}$ have independent columns with law $\calN(\boldsymbol{0},\matI_m/m)$. The goal is to choose $\vecx\in\{-1,1\}^n$ so that $\matM\vecx$ is not only entrywise positive, but bounded away from zero. The algorithm maintains a current state of the form $\vecs=\matM\vecx$ (and possibly a few extra bits of information). Upon seeing the next column $\veca$ of $\matM$ it decides on the corresponding sign in $\vecx$. The sign is effectively chosen to minimize the potential,
\[ 
    \Phi(\vecs)=\E_{I\sim \mathrm{Unif}([m])}\!\big[\psi(s_I)\big]
=\frac{1}{m}\sum_{i=1}^m \psi(s_i)\,,
\]
for a suitable ``energy function'' $\psi$, where the expectation is over a uniformly random coordinate $I\in[m]$. The wisdom of our choice $\psi(s) = e^{-\lambda s}$, where $\lambda = \sqrt{2/\pi}$, will become apparent shortly.

\noindent Given the current state $\vecs \in \bbR^m$, upon reading input $\veca \in \bbR^m$, the potential updates to,
\[ 
\Phi(\vecs') = \Phi(\vecs \pm \veca) = \E[\psi(s \pm a)] = \Phi(\vecs) \pm \E\left[\psi'(s)\,a\right] + \frac12 \E\left[\psi''(s)\,a^2\right] \pm \frac16 \E\left[\psi'''(s)\,a^3\right] + \cdots\,,
\]
assuming $\psi$ is analytic at $s$. All other expectations $\E[\cdot]$ in this section are over the randomness of the current Gaussian column $\veca$
(and any internal randomness of the algorithm), conditioned on the current state $\vecs$ unless stated otherwise. To simplify the analysis, we will choose the sign so as to minimize the leading term $\E[\psi'(s)a]$ instead of $\Phi(\vecs')$ itself.  That choice is $x_j = -\sign \E[\psi'(s)a]$.  The change in potential is,
\begin{equation}
\label{eq:expansion}
\Phi(\vecs') - \Phi(\vecs) = -\big\lvert{\E[\psi'(s)\,a]}\big\rvert + \frac12 \E[\psi''(s)\,a^2] + \sum_{i = 3}^\infty (\pm 1)^i \frac{1}{i!} \E[\psi^{(i)}(s)\,a^i]\,.
\end{equation}
To understand the typical behavior, we first analyze the expected change given $\vecs$. Heuristically, choosing the sign to minimize the linear term exploits the $1/\sqrt{m}$ fluctuations of $a$ to produce a systematic $1/m$ drop in $\Phi$, while the quadratic and higher-order terms contribute only smaller, controlled corrections.  As $\veca$ is $m$-dimensional normal with entrywise variance $1/m$, the key quantity $\E[\psi'(s)a]$ is a centered normal with variance $(1/m^2)\E[\psi'(s)^2]$. The leading term is the absolute value of this variable with a minus sign in front.  Its mean value is,
\[ 
    -\E\!\left[\left|\E[\psi'(s)\,a \mid \vecs]\right|\right]
    = -\sqrt{\frac{2}{\pi}}\,\frac{\sqrt{\E\left[\psi'(s)^2\right]}}{m}\,.
\]
The next term $\E[\psi''(s)\,a^2]/2$ contributes $+\frac{1}{2m}\E[\psi''(s)]$, since $\E[a^2]=1/m$ and $s$ is independent of $a$. For analytic $\psi$ with subexponential growth, the remainder is $O\!\left(\E[|\psi^{(4)}(s)|]/m^2\right)$ by \Cref{lem:gaussian-estimates}, using $\E|a|^k=\Theta(m^{-k/2})$. Intuitively, the linear term in the expansion enforces a negative drift (pulling the margins outward), while the quadratic term captures variance from the Gaussian increments and is controlled by moment bounds; the remaining higher-order terms are $O(1/m^2)$ for our choice of $\psi$. In summary, the average change in potential is,
\[
\E\!\left[\Phi(\vecs') - \Phi(\vecs)\,\middle|\,\vecs\right] =
\frac{1}{m}\biggl(-\sqrt{\frac{2}{\pi}} \sqrt{\E\left[\psi'(s)^2\right]} + \frac12  \E\left[\psi''(s)\right]
\biggr) + O\left(\frac{\E\left[\psi^{(4)}(s)\right]}{m^2}\right). 
\]
 Intuitively, the first term captures the outward pull from small margins, the second is a variance correction that stays strictly smaller at our choice of $\psi$, and the remainder vanishes at rate $1/m$. The potential $\psi$ should be decreasing to reward positivity.  To control outliers, it should be positive.  Thus its derivative $\psi'$ should be decreasing, and the second derivative $\psi''$ should be positive.  Then progress is made as long as the term $\sqrt{2/\pi} \sqrt{\E\left[\psi'(s)^2\right]}$ wins over $\E\left[\psi''(s)\right]$ in the battle of drifts.

For exponential potentials $\psi(x) = e^{-\lambda x}$, this is guaranteed as long as $0 < \lambda < \sqrt{8/\pi}$ because,
\[
\begin{aligned}
    &\sqrt{\frac{2}{\pi}} \sqrt{\E\left[\psi'(s)^2\right]}
    - \frac12 \E\left[ \psi''(s)\right] \\
    &\quad=
    \sqrt{\frac{2}{\pi}} \lambda \sqrt{\E\left[e^{-2\lambda s}\right]}
    - \frac12 \lambda^2 \E\left[e^{-\lambda s}\right]
    \geq
    (\sqrt{2/\pi} - \lambda / 2) \lambda \E\left[e^{-\lambda s}\right]\,,
\end{aligned}
\]
by convexity of moments.  The drift is optimized when $\lambda = \sqrt{2/\pi}$. In effect, the updates harvest many tiny $1/\sqrt{m}$ nudges into a uniform multiplicative decay of $\Phi$, which is exactly what ultimately clears the low-margin coordinates.

The algorithms below use three related routines. \(Push(m)\) is the basic one-step potential descent rule in an \(m\)-dimensional normalized state. \(Push(m,m_0)\) is the same rule when the active set has size \(m\) but the ambient Gaussian increments have variance \(1/m_0\). \(FocusedPush(m_0)\) repeatedly applies this scaled rule to the currently smallest coordinates, while updating all coordinates with the chosen signs.

It is useful to isolate the one-step sign rule used by \(Push\). For a state \(\vecs\in\mathbb R^m\) and a fresh column \(\veca\in\mathbb R^m\), define
\[
PushSign(\vecs,\veca)
=
\sign\!\left(\frac1m\sum_{i=1}^m a_i e^{-\lambda s_i}\right).
\]
The corresponding update is \(\vecs\gets\vecs+PushSign(\vecs,\veca)\veca\).

\begin{pseudocode}
{\bf Algorithm} \(Push(m)\) \\
{\bf State} $\vecs = (\dots, s, \dots) \in \mathbb{R}^m$ \\
On input $\veca \sim (1/\sqrt{m})\calN(\boldsymbol{0}, \matI_m)$, \\
{\small 1} \> Output \(x = PushSign(\vecs,\veca)\). \\
{\small 2} \> Update $\vecs$ to $\vecs + x\veca$. 
\end{pseudocode}

\noindent The average potential $\Phi_\lambda$ then satisfies the following recurrent inequality:

\begin{claim} 
\label{claim:recurrence}
When $\psi(s) = e^{-\lambda s}$, assuming $0 < \lambda \leq \sqrt{2/\pi}$,
\[ 
    \E\!\left[\Phi_\lambda(\vecs')\,\middle|\,\vecs\right] \leq \biggl(1 - \frac{\gamma \lambda}{m}\biggr)\Phi_\lambda(\vecs)\,,
\]
where $\gamma = \sqrt{2/\pi} - \sqrt{\pi/2}(\exp(1/\pi) - 1) > 0.328$.
\end{claim}
\begin{proof}
For the exponential potential, set \(w_i=e^{-\lambda s_i}\) and
\[
Z=\frac1m\sum_{i=1}^m w_i a_i,\qquad x=\sign(Z).
\]
The linear term in \(\Phi_\lambda(\vecs+x\veca)-\Phi_\lambda(\vecs)\) is \(-\lambda |Z|\).
We next justify that the adaptive sign does not increase the nonlinear remainder. Fix \(i\) and \(k\ge1\). Since \((a_i,Z)\) is jointly Gaussian and \(\mathrm{Cov}(a_i,Z)=w_i/m^2\ge0\), we may write
\[
a_i=\rho_i m^{-1/2}G+\sqrt{1-\rho_i^2}\,m^{-1/2}H,
\qquad
\sign(Z)=\sign(G),
\]
where \(\rho_i\in[0,1]\) and \(G,H\) are independent standard Gaussians. Expanding in powers of \(H\), only even powers of \(H\) survive, and \(\sign(G)G^{2\ell+1}=|G|^{2\ell+1}\), so
\[
\E[\sign(Z)a_i^{2k+1}\mid \vecs]\ge0.
\]
Thus all odd Taylor terms beyond the linear one have nonpositive contribution. Dropping them and summing the even terms gives
\[
\begin{aligned}
\E\!\left[\Phi_\lambda(\vecs')-\Phi_\lambda(\vecs)\,\middle|\,\vecs\right]
&\le
-\lambda\,\E[|Z|\mid\vecs]
+\frac1m\sum_{i=1}^m w_i
\sum_{k\ge1}\frac{\lambda^{2k}}{(2k)!}\E[a_i^{2k}]\\
&=
-\lambda\,\E[|Z|\mid\vecs]
+\left(\exp(\lambda^2/2m)-1\right)\Phi_\lambda(\vecs).
\end{aligned}
\]
The first term on the right hand side evaluates to,
\[
    -\lambda\,\E[|Z|\mid\vecs]
    =
    -\sqrt{\frac{2}{\pi}}\,\frac{\lambda}{m}\,
    \sqrt{\frac1m\sum_{i=1}^m e^{-2\lambda s_i}}
    \leq - \sqrt{\frac{2}{\pi}}\,\frac{\lambda}{m}\,\Phi_\lambda(\vecs) \,,
\]
by Cauchy--Schwarz.
Overall,
\[ 
\begin{aligned}
    \E\!\left[\Phi_\lambda(\vecs') - \Phi_\lambda(\vecs)\,\middle|\,\vecs\right]
    \leq
    \biggl(- \sqrt{\frac{2}{\pi}}\,\frac{\lambda}{m}
    + \exp\!\left(\lambda^2/2m\right) - 1\biggr)\Phi_\lambda(\vecs)\,.
\end{aligned}
\]
As $\lambda^2/2 \leq 1/\pi$ and $m \geq 1$, $\exp(\lambda^2/2m) - 1 \leq \pi(\exp(1/\pi) - 1)(\lambda^2/2m) \leq \sqrt{\pi/2}(\exp(1/\pi) - 1)(\lambda/m)$.  The desired bound follows.
\end{proof}

\begin{corollary}
\label{cor:supermart}
$\E\left[\Phi_\lambda(\vecs_t)\right] \leq (1 - \gamma \lambda /m)^t \Phi_\lambda(\vecs_0) \leq \exp\!\left(-\gamma\lambda t/m\right) \Phi_\lambda(\vecs_0)$, where $\vecs_t$ is the state after $t$ steps, assuming $\lambda^2 \leq 2/\pi$.
\end{corollary}

\begin{proof}
By~\Cref{claim:recurrence}, $\E[M_{t+1}\mid M_1,\ldots,M_t]\le M_t$, i.e.\ $M_t=(1 - \gamma\lambda /m)^{-t}\,\Phi_\lambda(\vecs_t)$ is a supermartingale. Applying the tower property of conditional expectation repeatedly gives $\E[M_t]\le\E[M_0]=M_0$.
\end{proof}

\begin{claim}\label{claim:push-tail}
Let \(\vecs_T\) be the state obtained by running \(Push(m)\) for \(T\) steps from an initial state \(\vecs_0\), with \(\lambda=\sqrt{2/\pi}\). Then,
\[
\E[\Phi(\vecs_T)]
\le
\exp\!\left(-\gamma'T/m\right)\Phi(\vecs_0).
\]
Consequently, for every threshold \(b\in\mathbb R\), the expected fraction of coordinates below \(b\) is at most,
\[
\exp\!\left(\lambda b-\gamma'T/m\right)\Phi(\vecs_0).
\]
\end{claim}

\begin{proof}
The first inequality is \Cref{cor:supermart}. If \(s_i<b\), then \(e^{-\lambda s_i}>e^{-\lambda b}\). Hence the fraction of coordinates below \(b\) is at most \(e^{\lambda b}\Phi(\vecs_T)\). Taking expectations gives the second bound.
\end{proof}

For example, from \Cref{claim:push-tail} with \(\vecs_0=\boldsymbol 0\) and \(b=0\), the expected fraction of negative entries after \(t\) steps is at most \(\exp\!\left(-\gamma't/m\right)\). Thus plain \(Push\) clears all negative coordinates with constant probability after \(O(m\log m)\) steps. Similarly, taking \(b=1\) shows that after \(O(m)\) steps all but a constant fraction of the coordinates have margin at least \(1\) in expectation. The limitation is that this still leaves open the possibility of a linear number of bad coordinates after only \(O(m)\) steps.

\subsection{A Constant Rate Algorithm}

The preceding analysis shows that plain \(Push\) steadily reduces the exponential potential, but by itself it gives only an inverse-logarithmic algorithmic rate if we insist that every coordinate becomes positive. The obstruction is that, after \(O(m)\) steps, a small but still linear set of low-margin coordinates may remain. The algorithm for constant algorithmic rate removes this loss by focusing the potential descent on the current low-margin tail.

In each stage, the entries of \(\vecs\) are split into critical and safe coordinates. The critical set consists of the \(m\) smallest entries at the current halving scale, and \(Push\) is applied only to this set when choosing the next sign. The same sign is then applied to all coordinates. Thus the algorithm spends most of its effort on the coordinates that still need help, while the safe coordinates are allowed to fluctuate.

As the size of the critical set changes from the initial $m_0$ to the ``current'' $m$, the potential function has to be scaled accordingly.  

\begin{pseudocode}
{\bf Algorithm} \(Push(m,m_0)\) \\
{\bf State} $\vecs\in\mathbb R^m$ \\
On input \(\veca\sim(1/\sqrt{m_0})\calN(\boldsymbol 0,\matI_m)\), \\
{\small 1} \> Set \(\lambda_{m,m_0}=\sqrt{2/\pi}\sqrt{m/m_0}\). \\
{\small 2} \> Set \(x\gets\sign(W)\), where \\
{\small 3} \> \> \(\displaystyle
W=
\frac1m\sum_{i=1}^m
\sqrt{\frac{m_0}{m}}\,a_i
\exp\!\left(-\lambda_{m,m_0}\sqrt{\frac{m_0}{m}}\,s_i\right)
\). \\
{\small 4} \> Update \(\vecs\gets\vecs+x\veca\).
\end{pseudocode}
Equivalently, this is the one-step rule of \(Push(m)\) applied to the rescaled state \(\sqrt{m_0/m}\vecs\) and rescaled input \(\sqrt{m_0/m}\veca\).
The state \(\vecs_t\) of \(Push(m, m_0)\) in step \(t\) satisfies
\begin{equation}
\label{eq:pushevolve}
\E\left[\Phi(\vecs_t)\right] \leq \exp\!\left(-\gamma' t/\sqrt{m m_0}\right)\,\Phi(\vecs_0).
\end{equation}

\begin{proof}[Proof of~\eqref{eq:pushevolve}]
Let $\lambda = \sqrt{2/\pi} \sqrt{m/m_0}$. By \Cref{cor:supermart}, 
\[
\E\left[\Phi_\lambda\left(\sqrt{m_0/m}\,\vecs_t\right)\right]
\leq
\exp\!\left(-\gamma' t/\sqrt{m m_0}\right)\,\Phi_\lambda\left(\sqrt{m_0/m}\,\vecs_0\right).
\]
By definition of $\Phi$, \(\Phi_\lambda\left(\sqrt{m_0/m}\,\vecs\right)
= \Phi_{\lambda\sqrt{m_0/m}}(\vecs)
= \Phi_{\sqrt{2/\pi}}(\vecs)\).
\end{proof}

\noindent For a vector $\vecx \in \mathbb{R}^{m_0}$ and subset of entries $S$, $\vecx|_S \in \mathbb{R}^S$ will denote the projection of $\vecx$ onto $S$. 

\noindent In the description and analysis of \(FocusedPush\), $m$ denotes the current scale in the halving sequence (ranging from $m_0$ down to $2$); the original output dimension $m_0=\alpha n$ is fixed throughout.

\noindent For notational clarity in the description of \(FocusedPush\), assume \(m_0\) is a power of two, so that the active set size can be halved cleanly. For general \(m_0\), one may round the active set sizes to integers; this changes only universal constants and none of the asymptotic conclusions.

\noindent We distinguish between the nonterminal scales \(m=m_0,m_0/2,\ldots,4\), where the proof maintains an invariant for the next critical set of size \(m/2\), and the terminal scale \(m=2\), where the remaining two critical coordinates are pushed above zero directly.

\begin{pseudocode}
{\bf Algorithm} \(FocusedPush(m_0)\) \\
{\bf State} $\vecs \in \mathbb{R}^{m_0}$ \\
{\small 1} \> For \(m=m_0,m_0/2,\ldots,2\), \\
{\small 2} \> \> Let $S$ be the coordinates of the $m$ smallest entries in $\vecs$. \\
{\small 3} \> \> For $t = \lceil 6\gamma^{-1}m_0^{5/8}m^{3/8}\rceil$ steps, \\
{\small 4} \> \> \> On input \(\veca\sim(1/\sqrt{m_0})\calN(\boldsymbol 0,\matI_{m_0})\), \\
{\small 5} \> \> \> \> Run \(Push(m,m_0)\) on \(\vecs|_S,\veca|_S\), obtaining sign \(x\). \\
{\small 6} \> \> \> \> Update all coordinates with the same sign: \(\vecs\gets\vecs+x\veca\).
\end{pseudocode}

\noindent Let \(L_{\mathrm{FP}}(m_0)\) denote the deterministic number of update steps used by \(FocusedPush(m_0)\) with the rounded loop counts above. Under the power-of-two convention,
\[
\begin{aligned}
    L_{\mathrm{FP}}(m_0)
    &=
    \sum_{k=0}^{\log_2 m_0-1}
    \left\lceil
    6\gamma^{-1}m_0^{5/8}(m_0 2^{-k})^{3/8}
    \right\rceil \\
    &\le
    \sum_{k=0}^{\log_2 m_0-1}
    \left(6\gamma^{-1}m_0\,2^{-3k/8}+1\right)
    \le Cm_0
\end{aligned}
\]
for all sufficiently large \(m_0\), where \(C<\infty\) is a universal constant.

\noindent For the rest of this section, $\lambda$ is fixed to $\sqrt{2/\pi}$.

\begin{theorem}[\(FocusedPush\)]
\label{thm:focusedpush}
Assuming $\Phi(\vecs_0) \leq \exp\!\left(-499\lambda\right)$, all entries of the final state of \(FocusedPush\) are positive with constant probability.
\end{theorem}

Call an entry $s$ of $\vecs$ \emph{$b$-safe} if $s > b$.  In the following two claims $\vecs_0$ and $\vecs_t$ are the initial and final state in iteration $m$ of loop {\small 1}.

\begin{claim}
\label{claim:motion}
Assuming $\Phi(\vecs_0|_S) \leq \exp\!\left(-\lambda(b - 1)\right)$, $\Phi(\vecs_t|_S) \leq \exp\!\left(-\lambda(b + 2)\right)$, except with probability at most $\exp\!\left(-3\lambda (m_0/m)^{1/8}\right)$.
\end{claim}
\begin{proof}
By~\eqref{eq:pushevolve}, 
\[ 
\begin{aligned}
\E\left[\Phi(\vecs_t|_S)\right]
&\leq \exp\!\left(-\lambda\left(b - 1 + 6(m_0/m)^{1/8}\right)\right) \\
&\leq \exp\!\left(-\lambda(b + 2)\right)\,\exp\!\left(-3\lambda (m_0/m)^{1/8}\right)\,.
\end{aligned}
\]
The conclusion follows from \Cref{lem:first-moment}.
\end{proof}

\begin{claim}
\label{claim:outliers}
For a nonterminal scale \(m\ge4\), assuming all entries of $\vecs_0$ outside $S$ are $b = 500(m/m_0)^{1/8}$-safe, (1) at most $m/4$ entries of $\vecs_t$ outside $S$ are not $500(m/2m_0)^{1/8}$-safe and (2) the unnormalized total potential of these unsafe entries is at most $m \exp\!\left(-\lambda (b + 2)\right)$ after the iteration, except with probability at most $0.007 \exp\!\left(-(m_0/m)^{1/8}\right)$.
\end{claim}
\begin{proof}
The sign chosen by \(Push\) in this iteration depends only on \(\veca|_S\). Conditional on that sign, the entries of \(\veca\) outside $S$ remain independent centered Gaussians, so the relevant entries are shifted by independent centered normals of variance \(t/m_0\), where
\[
6\gamma^{-1}m_0^{5/8}m^{3/8}
\le t
\le (6+\gamma)\gamma^{-1}m_0^{5/8}m^{3/8}.
\]
Let
\[
u=500\left(1-2^{-1/8}\right)\left(m/m_0\right)^{1/8}.
\]
For an entry not to be \(500(m/2m_0)^{1/8}\)-safe, the shift must take value at most \(-u\). By \Cref{lem:gaussian-estimates} and the upper bound on \(t\), for any given entry this probability is at most \(\exp\!\left(-44 (m_0/m)^{1/8}\right)\). By a union bound the probability that at least $m/4$ offending entries exist is at most, writing \(R=(m_0/m)^{1/8}\ge1\),
\begin{align*}
\binom{m_0}{m/4}\,\exp\!\left(-44R\right)^{m/4}
&\leq \left(4eR^8\,\exp\!\left(-44R\right)\right)^{m/4} \\
&\leq \left(4eR^8\,\exp\!\left(-40R\right)\right)^{1/4}\exp\!\left(-R\right) \\
&\leq 0.005 \exp\!\left(-R\right)\,.
\end{align*}
The average unnormalized total potential of the unsafe entries is at most, 
\[
    m_0\,\E\left[\exp\!\left(-\lambda X\right)\boldsymbol{1}\{X \leq b-u\}\right]\,,
\]
for a normal $X$ of mean $b$ and variance $t/m_0$. Applying Cauchy--Schwarz as
\[
\E\!\left[\exp\!\left(-\lambda X\right)\boldsymbol{1}\{X\le b-u\}\right]
\le
\sqrt{\E\!\left[\exp\!\left(-2\lambda X\right)\right]}\,\sqrt{\Pr(X\le b-u)},
\]
and using the upper bound on $t$ and the Gaussian tail estimate, this is at most
\begin{align*}
m_0 \sqrt{\E\!\left[\exp\!\left(-2\lambda X\right)\right]}\,\sqrt{\Pr(X \leq b-u)}
&\leq m_0\,\exp\!\left(-\lambda b\right)\,\exp\!\left(\lambda^2 t/m_0\right)\,\exp\!\left(-22R\right) \\
&\leq m_0\,\exp\!\left(-\lambda b\right)\,\exp\!\left(13R^{-3}\right)\,\exp\!\left(-22R\right) \\
&\leq m\,\exp\!\left(-\lambda (b + 2)\right)\,0.002\,\exp\!\left(-R\right).
\end{align*}
In the second line we used \(\lambda^2=2/\pi\), \(\gamma>0.328\), and \(t\le (6+\gamma)\gamma^{-1}m_0^{5/8}m^{3/8}\). The last inequality has numerical slack: for \(R\ge1\),
\[
R^8\exp\!\left(2\lambda+13R^{-3}-21R\right)\le 0.002.
\]
The claim follows from \Cref{lem:first-moment}.
\end{proof}

\begin{claim}
\label{claim:mainfocusedpush}
For a nonterminal scale \(m\ge4\), assume that all entries of $\vecs_0$ outside $S$ are $b(m)$-safe and $\Phi(\vecs_0|_S) \leq \exp\!\left(-\lambda(b(m) - 1)\right)$ before iteration $m$, with $b(m) = 500(m/m_0)^{1/8}$. Let $S'$ be the set of $m/2$ smallest entries of $\vecs$ after this iteration. Then, except with probability $\exp\!\left(-3\lambda(m_0/m)^{1/8}\right) + 0.007 \exp\!\left(-(m_0/m)^{1/8}\right)$, (1) all $b(m/2)$-unsafe entries are contained in $S'$, and (2) $\Phi(\vecs|_{S'}) \leq \exp\!\left(-\lambda\left(b(m/2)-1\right)\right)$.
\end{claim}
\begin{proof}
Assume the exceptional events from the above two claims do not occur. By \Cref{claim:motion}, at the end of iteration $m$, the total (unnormalized) potential in $S$ is at most $m\exp\!\left(-\lambda(b(m)+2)\right)$. An entry in $S$ is $b(m/2)$-unsafe only if $s_i < b(m/2) < b(m)$, hence its potential exceeds $\exp\!\left(-\lambda b(m)\right)$. If more than $me^{-2\lambda}$ entries of $S$ were $b(m/2)$-unsafe, their potential alone would exceed $me^{-2\lambda}\cdot\exp\!\left(-\lambda b(m)\right) = m\exp\!\left(-\lambda(b(m)+2)\right)$, contradicting the total potential bound; so at most $me^{-2\lambda}<m/4$ entries of $S$ are $b(m/2)$-unsafe. By~\Cref{claim:outliers}, at most $m/4$ entries of $\vecs_t$ outside $S$ are not $b(m/2)$-safe. Altogether there can be at most $m/2$ unsafe entries at the end of iteration $m$; since $S'$ is the set of the $m/2$ smallest entries, it contains all of them, establishing~(1). For~(2), $S'$ is contained in the union of the entries of $S$, the unsafe entries outside $S$, and $b(m/2)$-safe filler entries needed to bring the size to $m/2$. Their total potential is at most the contribution from $S$, at most $m\exp\!\left(-\lambda(b(m)+2)\right)$, plus the contribution of the unsafe entries outside $S$, also at most $m\exp\!\left(-\lambda(b(m)+2)\right)$ by \Cref{claim:outliers}, plus the potential of the safe filler entries, at most $(m/2)\exp\!\left(-\lambda b(m/2)\right)$. As $b(m)>b(m/2)$,
\[
\frac{m}{2}\,\Phi(\vecs|_{S'})
\;\le\;
2m\,\exp\!\left(-\lambda\left(b(m/2)+2\right)\right)
+\frac{m}{2}\,\exp\!\left(-\lambda b(m/2)\right)
\;\le\;
\frac{m}{2}\,\exp\!\left(-\lambda\left(b(m/2)-1\right)\right),
\]
where the last step uses \(4e^{-2\lambda}+1\le e^{\lambda}\) (numerically \(4\cdot0.20+1=1.80<e^{0.798}\approx2.22\)). Dividing by \(m/2\) gives \(\Phi(\vecs|_{S'})\le\exp\!\left(-\lambda\left(b(m/2)-1\right)\right)\).
\end{proof}

\begin{proof}[Proof of \Cref{thm:focusedpush}]
Apply \Cref{claim:mainfocusedpush} inductively over the halving scales \(m=m_0,m_0/2,\ldots,4\), with invariant: all $b(m)$-unsafe entries are contained in the active set $S$, and $\Phi(\vecs|_S)\le\exp\!\left(-\lambda(b(m)-1)\right)$. At $m=m_0$ the invariant holds: $S$ consists of all $m_0$ entries so the outside-$S$ condition is vacuous, and $\Phi(\vecs_0)\le\exp\!\left(-499\lambda\right)=\exp\!\left(-\lambda(b(m_0)-1)\right)$ by hypothesis. Before the final scale \(m=2\), the active set \(S\) of the two smallest coordinates contains all \(b(2)\)-unsafe entries and satisfies
\[
\Phi(\vecs|_S)\le \exp\!\left(-\lambda(b(2)-1)\right).
\]
The final push clears these remaining unsafe entries. \Cref{claim:motion} with \(m=2\) gives
\[
\Phi(\vecs_t|_S)\le \exp\!\left(-\lambda(b(2)+2)\right),
\]
except with probability at most \(\exp\!\left(-3\lambda (m_0/2)^{1/8}\right)\). Since \(|S|=2\), each of these two coordinates is at least \(b(2)+2-\lambda^{-1}\log 2>0\).

It remains to verify that the coordinates outside $S$ stay safe. At the final scale the sign depends only on \(\veca|_S\) with \(|S|=2\); since \(\veca|_S\) is independent of the entries of \(\veca\) outside $S$, the sign is independent of those increments, which are therefore independent centered Gaussians of variance \(t/m_0\), where
\[
t=\left\lceil 6\gamma^{-1}m_0^{5/8}2^{3/8}\right\rceil
\le (6+\gamma)\gamma^{-1}m_0^{5/8}2^{3/8}.
\]
Each coordinate starts above \(b(2)=500(2/m_0)^{1/8}\). Hence, by the same Gaussian tail estimate used in \Cref{claim:outliers} and a union bound, the probability that any coordinate outside $S$ falls below zero during the final scale is at most \(m_0\exp\!\left(-c m_0^{1/8}\right)\) for a universal constant \(c>0\). This is \(o(1)\), and is at most \(0.001\) for all sufficiently large \(m_0\). Thus, outside the exceptional events, every coordinate of the final state is positive.

The failure probability from the inductive scales and the final critical coordinates is at most the sum of,
\[
\exp\!\left(-3\lambda\right) + \exp\!\left(-3\lambda\,2^{1/8}\right) + \exp\!\left(-3\lambda\,2^{1/4}\right) + \cdots \leq 0.379\,,
\]
and
\[
0.007\left(\exp\!\left(-1\right) + \exp\!\left(-2^{1/8}\right) + \exp\!\left(-2^{1/4}\right) + \cdots\right) \leq 0.020\,.
\]
Together with the \(o(1)\) final-scale complement error, the success probability is bounded below by a positive universal constant for all sufficiently large \(m_0\).
\end{proof}

\subsection{Finding Collisions}
\label{subsec:finding-collisions}

We now adapt the margin-building procedure to collision finding. The algorithm first creates the desired Hamming separation by assigning opposite signs to \(\vecx\) and \(\vecy\) on an initial block of coordinates. After this block, every new coordinate receives the same sign in both vectors, so the overlap remains fixed. The relevant state is then the pointwise minimum,
\[
\vecs=\min(\matM\vecx,\matM\vecy),
\]
which evolves by the same update rule as a single margin vector whenever the two assignments use the same sign. Since \Cref{cor:supermart} applies even from a moderately large starting potential, we can run \(Push\) on this minimum state and then use \(FocusedPush\) for the final refinement.

\begin{pseudocode}
{\bf Algorithm} \(CollisionFinder(\matM,\theta;\beta_1)\) \\
{\bf Input} columns \(\veca_1,\ldots,\veca_n\) of \(\matM\in\mathbb R^{m\times n}\) \\
{\bf State} \(\vecs\in\mathbb R^m\), the pointwise minimum of the two current partial margin vectors, initially \(\vecs=\boldsymbol{0}\) \\[3pt]
{\small 1} \> Let \(L_{\mathrm{FP}}=L_{\mathrm{FP}}(m)\), and reserve the final \(L_{\mathrm{FP}}\) columns for refinement. \\
{\small 2} \> {\it Initialization:} \\
{\small 3} \> \> For $t=1,\ldots,\lfloor\beta_1 m\rfloor$, set $x_t\gets -1$ and $y_t\gets 1$. \\
{\small 4} \> \> Update \(\vecs\gets\min(\matM\vecx,\matM\vecy)\). \\[3pt]
{\small 5} \> {\it Drift phase:} \\
{\small 6} \> \> For $t=\lfloor\beta_1 m\rfloor+1,\ldots,n-L_{\mathrm{FP}}$, \\
{\small 7} \> \> \> On input column \(\veca_t\), set \(z_t\gets PushSign(\vecs,\veca_t)\). \\
{\small 8} \> \> \> Set \(x_t\gets z_t\) and \(y_t\gets z_t\), and update \(\vecs\gets\vecs+z_t\veca_t\). \\[3pt]
{\small 9} \> {\it Refinement:} \\
{\small 10} \> \> Update \(\vecs \gets \vecs - \theta \mathbf{1}\).\\
{\small 11} \> \> Run \(FocusedPush\) from \(\vecs\) on columns \(n-L_{\mathrm{FP}}+1,\ldots,n\), \\
\> \> using each output sign for both $\vecx$ and $\vecy$. \\[3pt]
{\bf Output} $(\vecx,\vecy)$.
\end{pseudocode}

In the ABP case $\theta=0$, if every entry of $\matM\vecx$ and $\matM\vecy$ is positive at termination then $\sign(\matM\vecx)=\sign(\matM\vecy)=\mathbf{1}_m$, so $(\vecx,\vecy)$ is a collision for ABP.

\begin{proof}[Proof of \cref{thm:algorithm}]
Let $\matB=\matA/\sqrt{\alpha}$ and set $\rho=\kappa/\sqrt{\alpha}$. Set \(D=\lceil \delta n/2\rceil\) and \(\beta_1=D/m\), and run \(CollisionFinder(\matB,\rho;\beta_1)\). Since the algorithm initializes for \(\lfloor\beta_1m\rfloor=D\) coordinates, the initial inner product is exactly \(\vecx^\top\vecy=n-2D\). For all sufficiently large \(n\), \(D\le n/2\), so \(\lvert\vecx^\top\vecy\rvert=n-2D\le(1-\delta)n\).

Let $\mathbf{s}$ be the pointwise minimum of $\matB\vecx$ and $\matB\vecy$ after these assignments. For each coordinate, the assigned part of $(\matB\vecx)_i$ is a Gaussian $G\sim\calN(0,\beta_1)$ and the assigned part of $(\matB\vecy)_i$ is $-G$, so $s_i=-|G|$. Hence,
\[
\E[\Phi(\mathbf{s})]
=
\E\!\left[\exp\!\left(\lambda |G|\right)\right]
\le 2\,\exp\!\left(\lambda^2\beta_1/2\right)\,,
\qquad \lambda=\sqrt{2/\pi}\,.
\]
	The parameter \(\beta_2\) below is used only in the analysis, as a lower bound on the number of drift steps available before the reserved refinement block. Take
	\[
	\beta_2
	=
	C\left(1+\beta_1+\rho\right),
	\]
	for a sufficiently large universal constant \(C\). Let \(L_{\mathrm{FP}}=L_{\mathrm{FP}}(m)\). The number of columns needed for initialization, \(\beta_2m\) drift steps, and the reserved refinement block is at most,
	\[
	    D+\beta_2m+L_{\mathrm{FP}}
	    \le
	    D+\beta_2m+Cm
	    \le
	    \left(\frac{\delta}{2}+\alpha(\beta_2+C)\right)n+1
	    \le
	    \left(\frac{\delta}{2}+C'\delta+C'\kappa\sqrt{\alpha}+C'\alpha\right)n,
	\]
	for a universal constant $C'$, with the additive \(+1\) absorbed for all sufficiently large \(n\). Choose $\delta_0>0$ small enough that $\delta/2+C'\delta<1/3$ for all $\delta<\delta_0$. Then choose $c>0$ small enough that $C'c<1/3$, and finally choose $\alpha_0(\delta)>0$ small enough that $C'\alpha<1/3$ for all $\alpha<\alpha_0(\delta)$. Since $\kappa\le c/\sqrt{\alpha}$, the number of used columns is at most $n$.
	Thus the actual drift phase in \(CollisionFinder\) has length at least \(\beta_2m\), because it runs through column \(n-L_{\mathrm{FP}}\). Applying \(Push\) for this many or more steps gives,
	\[
	    \E[\Phi] \le 2\,\exp\!\left(\lambda^2\beta_1/2 - \gamma'\beta_2\right)\,,
	\]
	where $\gamma'=\lambda\gamma>0$ is the fixed decay constant from \Cref{cor:supermart}. The choice of \(C\) above ensures,
	\[
	\E[\Phi] \le \frac12\,\exp\!\left(-499\lambda-\lambda\rho\right).
	\]
	Since the expectation is at most half of \(\exp\!\left(-499\lambda-\lambda\rho\right)\), \Cref{lem:first-moment} gives,
	\[
	\Pr\!\left[\Phi>\exp\!\left(-499\lambda-\lambda\rho\right)\right]\le \frac12.
	\]
	Thus, with probability at least \(1/2\), we have \(\Phi\le\exp\!\left(-499\lambda-\lambda\rho\right)\) before the refinement. During the drift phase and the refinement, every new column receives the same sign in $\vecx$ and $\vecy$, so the pointwise minimum evolves according to the same update rule since $\min(u+z a,v+z a)=\min(u,v)+z a$. The shift $\vecs \gets \vecs-\rho\mathbf{1}$ multiplies the potential by $e^{\lambda\rho}$, so after the shift we have $\Phi \le \exp\!\left(-499\lambda\right)\), and Theorem~\ref{thm:focusedpush} applies to the reserved final block.

	The algorithm assigns every one of the \(n\) coordinates: the initialization and drift phases assign the first \(n-L_{\mathrm{FP}}\) coordinates, and \(FocusedPush\) assigns the reserved final \(L_{\mathrm{FP}}\) coordinates. After initialization, every assigned coordinate receives the same sign in \(\vecx\) and \(\vecy\), so the final inner product remains exactly \(n-2D\le(1-\delta)n\). If \(c_{\mathrm{FP}}>0\) denotes the success probability in Theorem~\ref{thm:focusedpush}, then conditioning on the Markov event, Theorem~\ref{thm:focusedpush} gives success probability at least \(c_{\mathrm{FP}}\); hence the intersection has probability at least \(c_{\mathrm{FP}}/2\). On this event, \(CollisionFinder\) returns a pair $(\vecx,\vecy)$ satisfying,
	\[
	\min_i(\matB\vecx)_i\ge\rho,\qquad
	\min_i(\matB\vecy)_i\ge\rho,\qquad
\left\lvert\vecx^\top\vecy\right\rvert\le(1-\delta)\,n.
\]
Multiplying the first two inequalities by $\sqrt{\alpha}$ gives,
\[
\min_i(\matA\vecx)_i\ge\kappa,\qquad
\min_i(\matA\vecy)_i\ge\kappa,
\]
as required.
\end{proof}

\subsection{Simulations}

This subsection is purely empirical and is not used in any of the formal results below.

We implemented a version of the \(CollisionFinder\) algorithm to explore the empirical transition. The implementation processes columns of $\matA \in \mathbb{R}^{m \times n}$ online, uses the one-sided potential $\Phi(s)=\frac{1}{m}\sum_{i=1}^m e^{-\lambda s_i/\sqrt{\alpha}}$ with $\lambda=\sqrt{2/\pi}$, and applies the drift-based decision rule described in Section~\ref{sec:evolution}. All experiments use $\delta = 0.1$ and Gaussian matrices $\matA_{ij} \sim \mathcal{N}(0, 1/n)$. The implementation includes practical choices, so success probabilities should be interpreted as empirical evidence for the scaling rather than as a direct verification of the proof.

\begin{figure}[!ht]
  \centering
  \begin{subfigure}[b]{0.48\textwidth}
    \centering
    \includegraphics[width=\textwidth]{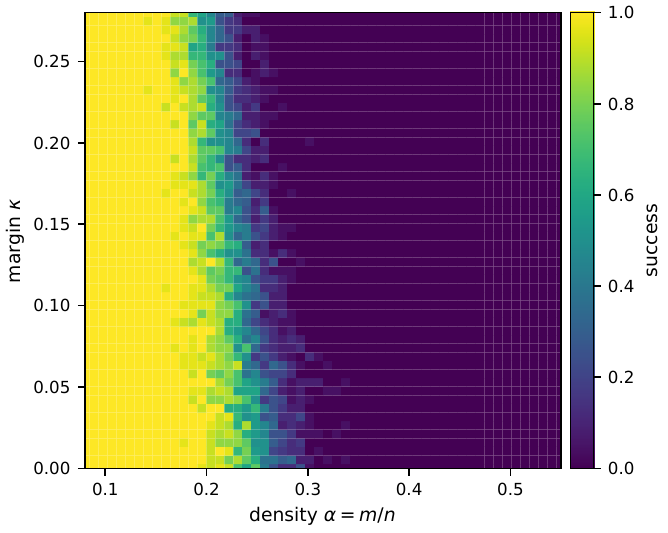}
  \end{subfigure}
  \hfill
  \begin{subfigure}[b]{0.48\textwidth}
    \centering
    \includegraphics[width=\textwidth]{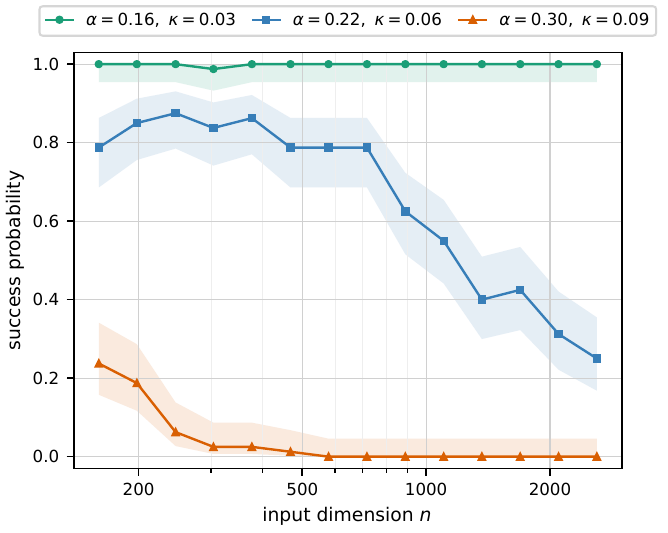}
  \end{subfigure}
  \caption{Empirical behavior of the online collision heuristic. Left: Phase diagram on a $54 \times 54$ grid with $n=420$. The observed transition is consistent with dependence on $\kappa\sqrt{\alpha}$. Right: Success probability versus $n$ for three configurations. The easy regime $(\alpha,\kappa)=(0.16,0.03)$ remains near $1$, the transition regime $(0.22,0.06)$ decreases with $n$, and the low-success regime $(0.30,0.09)$ quickly decays to $0$.}
  \label{fig:simulations}
\end{figure}

Figure~\ref{fig:simulations} shows two empirical patterns. The phase diagram shows the success probability over a $54 \times 54$ grid spanning $\alpha \in [0.08, 0.55]$ and $\kappa \in [0, 0.28]$ for $n=420$. While the exact transition curve fluctuates mildly due to implementation choices, a clear inverse relationship between $\alpha$ and $\kappa$ emerges: the algorithm succeeds almost always in the lower-left region, where $\kappa\sqrt{\alpha}$ is smaller, and almost never in the upper-right region, where $\kappa\sqrt{\alpha}$ is larger. This diagonal pattern is consistent with the scaling in Theorem~\ref{thm:algorithm}, suggesting that the empirical boundary is controlled by the dimensionless parameter $\kappa \sqrt{\alpha}$.

To test whether this empirical transition persists as $n$ grows, we ran the algorithm for $n \in [160, 2600]$ across three representative configurations. In the easy regime $(\alpha,\kappa)=(0.16,0.03)$, the success probability remains near $1$ across all $n$. In the transition regime $(0.22,0.06)$ it decreases with $n$, while in the low-success regime $(0.30,0.09)$ it quickly decays toward $0$.

Overall, while precise probabilities are sensitive to implementation details, the experiments support the prediction that the empirical transition is controlled by the product $\kappa\sqrt{\alpha}$.

\section{Establishing OGP}
\label{sec:random-activation}

In this section we analyze a randomized activation function. For each coordinate $i\in[m]$ we fix an odd integer $K\geq 1$ and draw $K$ i.i.d. thresholds $t_{ij}$ uniformly at random from the range $[-\kappa, \kappa]$, for $j\in[K]$, and define the coordinate-wise activation,
\[
    \varphi_i(x)\,=\,\prod_{j=1}^K \mathrm{sign}(x-t_{ij})\,,\qquad \varphi(\matA \vecx)_i\,=\,\varphi_i\left((\matA \vecx)_i\right)\,.
\]
The thresholds are independent across $i$ and $j$, and independent of $\matA$. Note that the activation is constant on each tail: for $x>\kappa$ all factors are $+1$, while for $x<-\kappa$ all factors are $-1$, so the product is $(-1)^K=-1$ for odd $K$. We restrict to odd \(K\) for simplicity; the oddness is used in the row-wise contraction argument below, and we do not analyze the even-\(K\) case.
The proof is a first-moment argument over admissible \(r\)-tuples of collision pairs. For a fixed row, admissibility and a selection lemma produce \(r+1\) Gaussian projections whose covariances are separated from one. The random thresholds then give a row-wise contraction, summarized by the exponent \(\Lambda_{r,\delta}(\kappa,K)\). Combining this row bound with entropy counting and independence across rows proves \Cref{thm:randomized-ogp-first-moment}. The ensemble version repeats the row-wise argument for correlated-prefix instances, where admissibility is controlled by \(1-t/n\), and the online lower bound follows by coupling \(r\) runs of an online algorithm and symmetrizing the suffix signs.

Fix $r\ge2$ and $\delta\in(0,1]$, and let $2 < c_1 < c_2 < 2c_1$ be constants. Throughout this section, we will let $(X_i,Y_i)_{i=1}^r$ be a set of centered Gaussian variables with unit variances. Let $Z_i, Z_j$ be two random variables from the set. For Gaussian projections arising from Boolean vectors in a common instance,
\[
d(\langle a,u\rangle,\langle a,v\rangle)
=\frac{1-\mathrm{Cov}(\langle a,u\rangle,\langle a,v\rangle)}{2}
=\frac{1}{n}\Delta_{\mathrm{Ham}}(u,v),
\]
and hence \(d\) satisfies the triangle inequality on the variables considered below.
More generally, the selection argument below is used only for Gaussian families for which this
quantity defines a metric.

Let $a,b,\tau > 0$ be constants. A set of centered Gaussian variables $(X_i,Y_i)_{i=1}^r$ with unit variances is said to be an \emph{admissible set} if it satisfies the following two conditions:
\begin{enumerate}
    \item \emph{(Internal overlap).} For every $i$, $d(X_i, Y_i) \geq \tau$.
    \item \emph{(External overlap).} For every $i \neq j$, $a < d(X_i,X_j) + d(Y_i, Y_j) < b$.
\end{enumerate}
For the Gaussian projections of a $\delta$-admissible family from Definition~\ref{def:ogp}, we have,
\[
\tau=\delta/2,\qquad a=\delta/c_2,\qquad b=\delta/c_1.
\]
Conditioned on a pair $(X,Y)$, a single threshold falls between them with probability $\Theta(|X-Y|/\kappa)$ while both lie in $[-\kappa,\kappa]$; $K$ independent thresholds then amplify this effect multiplicatively. Let
\[
\eta_{\mathrm{si}}=\min\left\{1-\frac2{c_1},\,\frac12\left(\frac2{c_2}-\frac1{c_1}\right)\right\},
\qquad
\eta_{\mathrm{ens}}=\min\left\{\frac1{c_2},\,1-\frac1{c_1},\,1\right\},
\qquad
\eta_0=\frac12\min\{\eta_{\mathrm{si}},\eta_{\mathrm{ens}}\}.
\]
Thus the selection lemma below, applied with \(\gamma=\eta_0\delta/2\), gives covariance separation at least \(\eta_0\delta\), and the same \(\eta_0\) is no larger than the separated-subfamily constants used in the ensemble argument. Let \(\bar c>0\) be the absolute contraction constant from Claim~\ref{claim:main rogp}. Fix constants \(c_0,C_0,C_1>0\), depending only on \(c_1,c_2\), with \(c_0=\bar c\sqrt{\eta_0}\) and \(C_0,C_1\) large enough for Lemma~\ref{lem:uniform-contraction}, and set
\[
p_{\max}=\min\left\{1,\frac{\kappa}{c_0\sqrt{\delta}}\right\}.
\]
For compactness, define
\[
\Lambda_{r,\delta}(\kappa,K)
=
-\log_2\!\left(
2^{-r}
+
\inf_{p\in(0,p_{\max}]}
\left[
\left(1-c_0p\sqrt{\delta}/\kappa\right)^K
+ C_0\,r^2p
+ C_1\,r\,e^{-\kappa^2/2}
\right]
\right).
\]
The main result of this section is the following.

\begin{theorem}[First-Moment OGP Bound]\label{thm:randomized-ogp-first-moment}
Fix constants $r\ge2$, $\delta\in(0,1]$, and $2<c_1<c_2<2c_1$, a parameter $\kappa>0$, an odd integer $K\ge1$, and $\alpha\in(0,1]$.
Let $\mathsf{Coll}_{r,\delta}=|S(r,\delta,c_1,c_2;\matA)|$ be the number of $\delta$-admissible $r$-tuples of collisions (as in \Cref{def:ogp}). The expectation $\E[\cdot]$ is over the draw of the Gaussian matrix $\matA$ and the randomized thresholds of $\varphi$.
Then,
\begin{equation}\label{eq:log-firstmoment}
    \frac{1}{n}\log_2 \E\!\left[\mathsf{Coll}_{r,\delta}\right]
    \le
    2 + 2(r-1)H\!\left(\delta/(2c_1)\right)
    - \alpha\,\Lambda_{r,\delta}(\kappa,K)
    + o(1).
\end{equation}
\end{theorem}

\noindent In particular, once
\[
\alpha\Lambda_{r,\delta}(\kappa,K)
\]
exceeds \(2+2(r-1)H(\delta/(2c_1))\) by a fixed positive constant, the exponent in~\eqref{eq:log-firstmoment} is negative and the first moment of $r$-collisions vanishes exponentially in $n$.

\begin{corollary}[Single-Instance $r$-OGP]\label{corollary:rogp}
There exist constants $c>0$ and $C<\infty$, depending only on $c_1,c_2$, such that the following holds for every \(\alpha\in(0,1]\), \(\kappa>0\), and odd integer \(K\ge1\).
Assume $H(\delta/(2c_1)) \leq c\alpha$, $\alpha r \geq C$, and $\Lambda_{r,\delta}(\kappa,K)\ge (1-c)r$. Then no $r$-collisions exist except with probability $\exp(-\Omega(\alpha r n))$. Consequently, with probability at least $1-2^{-\Omega(n)}$, the randomized activation function $\varphi$ satisfies $r$-OGP.
\end{corollary}
\begin{proof}
By \Cref{thm:randomized-ogp-first-moment}, the first-moment exponent is at most
\[
    2 + 2(r-1)H(\delta/(2c_1)) - \alpha\Lambda_{r,\delta}(\kappa,K) + o(1) \leq -\Omega(\alpha r)\,.
\]
The expected number of $r$-collisions is then at most $\exp(-\Omega(\alpha r n))$.  By \Cref{lem:first-moment}, no $r$-collision exists except with such probability.
\end{proof}

The quantity \(\Lambda_{r,\delta}(\kappa,K)\) is the row-wise collision exponent: the proof below shows that a fixed admissible \(r\)-tuple collides on one row with probability at most \(2^{-\Lambda_{r,\delta}(\kappa,K)}\). A simple sufficient regime is obtained by choosing \(p=\Theta(2^{-r}/r^2)\): for constants \(c>0\) and \(C<\infty\) depending only on \(c_1,c_2\), the conditions
\[
\kappa^2\ge C(r+\log r),
\qquad
K\text{ is odd and }K\ge C\,2^r r^3\,\frac{\kappa}{\sqrt{\delta}}
\]
imply
\[
\Lambda_{r,\delta}(\kappa,K)\ge (1-c)r.
\]
The extra factor of \(r\) in the lower bound on \(K\) comes from making \(\left(1-c_0p\sqrt{\delta}/\kappa\right)^K\) comparable to \(2^{-r}\), while the choice \(p=\Theta(2^{-r}/r^2)\) keeps the \(C_0\,r^2p\) error at the same scale.

The proof of \Cref{thm:randomized-ogp-first-moment} has two ingredients. First, the overlap constraints give the entropy count below. Second, for every fixed admissible family, the Gaussian variables seen by one row contain a large separated subfamily; this forces a uniform contraction under the randomized thresholds. Independence across rows then gives the stated first moment.

\begin{lemma}[Entropy Count for Admissible Families]\label{lem:admissible-entropy-count}
For fixed constants $r\ge2$, $\delta\in(0,1]$, and $2<c_1<c_2<2c_1$,
\[
\frac1n\log_2\left|T(r,\delta,c_1,c_2)\right|
\le
2+2(r-1)H(\delta/(2c_1))+o(1).
\]
\end{lemma}
\begin{proof}
The first pair $(\vecx_1,\vecy_1)$ has at most $2^{2n}$ choices. For each $i>1$, the \emph{lower} external-overlap bound $(\vecx_i^\top\vecx_1+\vecy_i^\top\vecy_1)/(2n)\ge1-\delta/c_1$, combined with $u^\top v=n-2\Delta_{\mathrm{Ham}}(u,v)$, gives
\[
\Delta_{\mathrm{Ham}}(\vecx_i,\vecx_1)+\Delta_{\mathrm{Ham}}(\vecy_i,\vecy_1)
\le \frac{\delta n}{c_1}.
\]
Thus, after fixing $(\vecx_1,\vecy_1)$, the logarithm of the number of choices for $(\vecx_i,\vecy_i)$ is at most
\[
\log_2\sum_{\ell\le \delta n/c_1}\binom{2n}{\ell}
\le
\left(2H(\delta/(2c_1))+o(1)\right)n.
\]
Multiplying this bound over $i=2,\ldots,r$ proves the claim. The internal-overlap condition and the remaining external-overlap conditions can only reduce the count.
\end{proof}

\begin{claim}\label{claim:main rogp}
Let $Z_1,\dots,Z_{2s}$ be centered Gaussian random variables with unit variances.
Suppose that among them there are $s+1$ variables such that, for every distinct pair
$Z_i,Z_j$ in this subfamily,
\[
\mathrm{Cov}(Z_i,Z_j)\le 1-\eta .
\]
Then there is an absolute constant \(\bar c>0\) such that, with $\varphi$ the randomized activation function and for every \(p\in(0,1)\) with \(\bar c p\sqrt{\eta}/\kappa\le1\),
\begin{align*}
    \E\left[\prod_{i=1}^{2s} \varphi(Z_i)\right]
    \leq
    \left(1 - \bar c p\sqrt{\eta}/\kappa\right)^K
	    + \binom{s+1}{2}\,p
	    + 2s\,e^{-\kappa^2/2}.
\end{align*}
\end{claim}
\begin{proof}
By a union bound and \Cref{lem:gaussian-estimates}, $\abs{Z_i} \leq \kappa$ for all $i$ except with probability at most $2s\,e^{-\kappa^2/2}$. By another union bound and \Cref{lem:gaussian-estimates}, since each selected difference has variance at least \(2\eta\), $\abs{Z_i - Z_j} \geq p\sqrt{\eta} =: q$ for all selected pairs except with probability at most $\binom{s + 1}{2}p$.

On this good event all \(Z_i\) lie in \([-\kappa,\kappa]\), and the selected \(s+1\) variables are pairwise separated by at least \(q\). We show that $\sign(Z_1- T)\cdots\sign(Z_{2s}-T)$ takes value $-1$ with probability at least $q/2\kappa$ for a threshold $T$ chosen uniformly from $[-\kappa, \kappa]$.

Let $\pi$ order $Z_1,\ldots,Z_{2s}$ increasingly. For any threshold $T\in(Z_{\pi(j)},Z_{\pi(j+1)})$, exactly $j$ of the quantities $Z_i-T$ are negative, so the sign product equals $(-1)^j$ on this interval. Since all \(Z_i\in[-\kappa,\kappa]\), these adjacent intervals are not clipped by the threshold range. Thus, up to endpoints of measure zero, the set of thresholds for which the product equals $-1$ is precisely the union of the odd-parity intervals
\[
(Z_{\pi(1)},Z_{\pi(2)}),\ (Z_{\pi(3)},Z_{\pi(4)}),\ \ldots,\ (Z_{\pi(2s-1)},Z_{\pi(2s)}),
\]
and
\[
\Pr_T\!\left[\prod_{i=1}^{2s}\mathrm{sign}(Z_i-T)=-1\right]
=\frac{1}{2\kappa}\sum_{j=1}^{s}
(Z_{\pi(2j)}-Z_{\pi(2j-1)}).
\]

If every odd-parity interval had length less than \(q\), then each pair of ordered positions \((\pi(2j-1),\pi(2j))\) could contain at most one selected variable, because selected variables are pairwise separated by at least \(q\). These \(s\) pairs partition the \(2s\) positions, so they would contain at most \(s\) selected variables in total, contradicting the existence of \(s+1\) selected variables. Hence some odd-parity interval has length at least \(q\), and it contributes at least \(q/2\kappa\) to the probability above.

In summary, $\sign(Z_1- T)\cdots\sign(Z_{2s}-T)$ takes value $-1$ with probability at least \(q/(2\kappa)\), and hence its one-threshold expectation satisfies
\[
g := \E_T\!\left[\prod_{i=1}^{2s}\sign(Z_i-T)\right] \le 1-q/\kappa\le 1-\bar c p\sqrt{\eta}/\kappa
\]
for a fixed absolute \(\bar c>0\). Here \(g=g(Z_1,\ldots,Z_{2s})\) is a function of the Gaussian values. Since the $K$ thresholds $T_1,\ldots,T_K$ are independent of each other and of the $Z$'s, the conditional expectation over thresholds given the $Z$ values satisfies
\[
\E\!\left[\prod_{i=1}^{2s}\varphi(Z_i)\,\middle|\,Z_1,\ldots,Z_{2s}\right]
=\E\!\left[\prod_{j=1}^K\prod_{i=1}^{2s}\sign(Z_i-T_j)\,\middle|\,Z_1,\ldots,Z_{2s}\right]
=g(Z_1,\ldots,Z_{2s})^K.
\]
On the good event, \(g\le 1-\bar c p\sqrt{\eta}/\kappa\). Since \(K\) is odd the map \(x\mapsto x^K\) is increasing, so on the good event \(g^K\le(1-\bar c p\sqrt{\eta}/\kappa)^K\). Using \(|g|\le1\) on the bad event, the bound on the unconditional expectation follows from the law of total probability.
\end{proof}

\begin{lemma}[Selection Lemma]
\label{lem:r+1 tuple exists}
Let $(X_i,Y_i)_{i=1}^r$ be an admissible set, and assume that $d(Z,Z')=(1-\mathrm{Cov}(Z,Z'))/2$ is a metric on $\{X_i,Y_i:i\in[r]\}$. If $\tau > b$ and $2a > b$, then for every $\gamma < \min\{\tau-b, (2a-b)/4\}$, there exists a subset of $r+1$ variables $Z_1, \ldots Z_{r+1}$ satisfying $d(Z_i, Z_j) \geq \gamma$.
\end{lemma}
\begin{proof}
    Fix any $\gamma < \min\{\tau-b, (2a-b)/4\}$, and define the set,
    \[
    S = \{i \in [r] : d(X_i, X_j) \geq \gamma \text{ for all } j \neq i\}\,.
    \]
    To show the statement we condition on the set $S$.
    \begin{itemize}
        \item If $S=\{1,\ldots, r\}$, then we may take all $X_i$ and one arbitrary $Y_j$. The statement follows since for all $i,j$,
\[
d(X_i, Y_j) \ge d(X_i, Y_i) - d(Y_i, Y_j) \ge \tau - b > \gamma,
\]
where we used the triangle inequality and the overlap conditions.
        \item If $S \neq \{1, \ldots r\}$, we take all $X_i$ for $i \in S$, all $Y_j$ for $j \not\in S$, and an arbitrary $X_{j^*}$ for $j^* \not\in S$. By definition of $S$, the condition is satisfied for any pairs of $X$'s. Distances between $X$'s and $Y$'s are again bounded by $d(X_i, Y_j) \geq \tau - b$. It remains to bound the distances between the $Y$'s.

We claim that if $d(X_i, X_j) < \gamma$, then $d(Y_j, Y_k) \geq \gamma$ for every $k\neq j$. For $k=i$, the external lower bound gives $d(Y_j,Y_i)\ge a-d(X_i,X_j)>a-\gamma>\gamma$, since $\gamma<a/2$. For distinct $i,j,k$, we will show it is impossible to have both $d(X_i, X_j) < \gamma$ and $d(Y_j, Y_k) < \gamma$, provided $4\gamma < 2a-b$. Assume for contradiction it is possible. By the external overlap condition, we get $d(Y_i, Y_j) \geq a - \gamma$. By the triangle inequality, $d(Y_i, Y_k) \geq a - 2\gamma$. From external overlap, we obtain $d(X_i, X_k) \leq b - a + 2\gamma$. By triangle inequality again, we get $d(X_j, X_k) \leq b- a+3\gamma$. Adding $d(Y_j, Y_k) < \gamma$, we finally get,
    \[
        d(X_j, X_k) + d(Y_j, Y_k) < b-a+4\gamma\,,
    \]
    which contradicts the lower bound $a \le d(X_j, X_k) + d(Y_j, Y_k)$. Therefore, if $d(X_i, X_j) < \gamma$, then $d(Y_j, Y_k) \geq \gamma$ for all $k \neq j$. It follows that all selected $Y$'s are pairwise at distance at least $\gamma$, and hence all pairwise distances in the chosen set are at least $\gamma$, completing the proof.\qedhere
    \end{itemize}
\end{proof}

\begin{lemma}[Uniform Row Contraction]
\label{lem:uniform-contraction}
There exist constants $c_0,C_0,C_1>0$, depending only on $c_1,c_2$, such that the following holds. Let $(X_i,Y_i)_{i=1}^r$ be an admissible set, and assume that \(d(Z,Z')=(1-\mathrm{Cov}(Z,Z'))/2\) is a metric on \(\{X_i,Y_i:i\in[r]\}\). Then in the randomized oscillating activation model with odd \(K\), for every $p\in(0,p_{\max}]$,
\[
\Pr[\varphi(X_j)=\varphi(Y_j)\ \forall j]
\le
2^{-r}
+ \left(1-c_0p\sqrt{\delta}/\kappa\right)^K
+ C_0\,r^2p
+ C_1\,r\,e^{-\kappa^2/2}.
\]
\end{lemma}

\begin{proof}
We arithmetize the collision probability as follows:
\[
\Pr[\varphi(X_j)=\varphi(Y_j)\ \forall j]
= \mathbb{E}\!\left[\prod_{j=1}^r \frac{1+\varphi(X_j)\varphi(Y_j)}{2}\right]
=2^{-r}\sum_{S\subseteq[r]}\mathbb{E}\!\left[\prod_{j\in S}\varphi(X_j)\varphi(Y_j)\right].
\]
The term \(S=\emptyset\) is \(1\). Fix any nonempty \(S\) of size \(s\). The restricted family \((X_j,Y_j)_{j\in S}\) remains admissible with the same parameters \(\tau,a,b\). By Lemma~\ref{lem:r+1 tuple exists}, applied to this restricted family with \(\gamma=\eta_0\delta/2\), there exist \(s+1\) variables among \(\{X_j,Y_j : j\in S\}\) such that all pairwise covariances are at most \(1-\eta_0\delta\). Applying \Cref{claim:main rogp} with \(\eta=\eta_0\delta\) yields
\[
	\mathbb{E}\!\left[\prod_{j\in S}\varphi(X_j)\varphi(Y_j)\right]
	\le \left(1 - c_0p\sqrt{\delta}/\kappa\right)^K + C_0\,s^2p + C_1\,s\,e^{-\kappa^2/2},
	\]
where \(c_0=\bar c\sqrt{\eta_0}\), and \(C_0,C_1\) absorb the absolute constants in \Cref{claim:main rogp}. Since \(s\le r\), each nonempty term is at most
\[
\left(1 - c_0p\sqrt{\delta}/\kappa\right)^K + C_0\,r^2p + C_1\,r\,e^{-\kappa^2/2}.
\]
Therefore
\[
\Pr[\varphi(X_j)=\varphi(Y_j)\ \forall j]
\le
2^{-r}
+
\left(1 - c_0p\sqrt{\delta}/\kappa\right)^K
+ C_0\,r^2p
+ C_1\,r\,e^{-\kappa^2/2},
\]
which is the claimed bound. Taking the infimum over \(p\in(0,p_{\max}]\) gives the row bound \(2^{-\Lambda_{r,\delta}(\kappa,K)}\) used below.
\end{proof}

\begin{proof}[Proof of \Cref{thm:randomized-ogp-first-moment}]
Consider a fixed $\delta$-admissible family of $r$ pairs $\{(\vecx_i,\vecy_i)\}_{i=1}^r$ and a fixed row $\veca$ of $\matA$, where $\veca\sim\mathcal N(\mathbf 0,\matI/n)$. For each $i\in[r]$, let
\[
X_i=\langle \veca,\vecx_i\rangle,\qquad Y_i=\langle \veca,\vecy_i\rangle.
\]
Then $(X_i,Y_i)_{i=1}^r$ is a jointly Gaussian family with unit variances. Moreover, since $\{(\vecx_i,\vecy_i)\}_{i=1}^r$ is $\delta$-admissible in the sense of Definition~\ref{def:ogp}, the Gaussian family $(X_i,Y_i)_{i=1}^r$ is admissible in the above sense, with
\[
\tau=\delta/2,\qquad a=\delta/c_2,\qquad b=\delta/c_1.
\]
	For Gaussian projections from Boolean vectors, \(d(\langle \veca,u\rangle,\langle \veca,v\rangle)=\Delta_{\mathrm{Ham}}(u,v)/n\), so the metric assumption in Lemma~\ref{lem:uniform-contraction} holds. Since $2<c_1<c_2<2c_1$, Lemma~\ref{lem:r+1 tuple exists} applies, and hence so does Lemma~\ref{lem:uniform-contraction}. Therefore the single-row collision probability is bounded by
\[
\Pr\!\left[\varphi(X_i)=\varphi(Y_i)\ \forall i\in[r]\right]
\le 2^{-\Lambda_{r,\delta}(\kappa,K)}.
\]

Since the rows of $\matA$ are independent, the probability that this fixed $\delta$-admissible family collides across all $m=\alpha n$ rows is at most
\[
2^{-\alpha n\Lambda_{r,\delta}(\kappa,K)}.
\]
Therefore,
\[
\log_2 \E[\mathsf{Coll}_{r,\delta}]
\le
\log_2\left|T(r,\delta,c_1,c_2)\right|
-\alpha n\Lambda_{r,\delta}(\kappa,K).
\]
The desired bound follows after applying \Cref{lem:admissible-entropy-count} and dividing by \(n\).
\end{proof}

\subsection{\texorpdfstring{Ensemble $r$-OGP}{Ensemble r-OGP}}\label{subsec:ensemble-rogp}

The $r$-OGP established above concerns configurations within a single random instance.
For the hardness argument, we require the corresponding \emph{ensemble} $r$-OGP,
which forbids such configurations across a correlated family of instances that share a common prefix and have independent suffixes: for any $t\in\{0,\ldots,n\}$ consider $r$ correlated instances that share the first $t$ columns and have independent suffixes (the randomized thresholds of the activation, when present, are shared across instances). 

\begin{definition}[Ensemble $r$-OGP]\label{def:ensemble-rogp}
Fix $r\ge 2$ and $\delta\in(0,1]$, and let $(c_1,c_2)$ be the forbidden external-overlap
window from \Cref{def:ogp}.
For $t\in\{0,1,\ldots,n\}$, let $\mathsf{A}^{(t)}=(\matA^{(1)},\ldots,\matA^{(r)})$ be the correlated ensemble in which the first $t$ columns coincide across all $\matA^{(j)}$ and the remaining $n-t$ columns are independent (with the same randomized thresholds shared across all instances, when present). A family $\{(\vecx^{(j)},\vecy^{(j)})\}_{j=1}^r\subseteq(\{-1,1\}^n)^2$ is called \emph{$t$-prefix aligned} if
\[
\vecx^{(1)}_{\le t}=\cdots=\vecx^{(r)}_{\le t}
\qquad\text{and}\qquad
\vecy^{(1)}_{\le t}=\cdots=\vecy^{(r)}_{\le t}.
\]
It is called \emph{ensemble $\delta$-admissible} if it satisfies the following conditions.
\begin{enumerate}
\item \emph{(Internal Overlap)}. For every $j\in[r]$,
$\frac{1}{n}\left|\langle \vecx^{(j)},\vecy^{(j)}\rangle\right| \le 1-\delta$;
\item \emph{(External Overlap)}. For every distinct $j\neq j'$,
\[
1 - \frac{\delta}{c_1} \le \frac{1}{2n}(\langle \vecx^{(j)},\vecx^{(j')}\rangle+\langle \vecy^{(j)},\vecy^{(j')}\rangle) \le 1 - \frac{\delta}{c_2}.
\]
\end{enumerate}

Let $\mathsf{Coll}^{\mathrm{ens}}_{r,\delta}(t)$ denote the number of
$t$-prefix aligned, ensemble $\delta$-admissible families such that each
pair is a collision on its own instance:
\[
\varphi(\matA^{(j)}\vecx^{(j)})
=
\varphi(\matA^{(j)}\vecy^{(j)})
\qquad\text{for all }j\in[r].
\]
We say \emph{ensemble $r$-OGP holds at $(\alpha,\delta)$} if, for every fixed \(t\in\{0,1,\ldots,n\}\) satisfying
\[
\frac{\delta}{c_2}<1-\frac{t}{n}<\frac{\delta}{c_1},
\]
the corresponding correlated-prefix ensemble satisfies
\[
\Pr_{\mathsf A^{(t)}}\!\left[\mathsf{Coll}^{\mathrm{ens}}_{r,\delta}(t)=0\right]
\ge 1-2^{-\Omega(n)}.
\]
This is a fixed-\(t\) statement, uniform over choices of \(t\) in the displayed window; it does not assert a simultaneous union bound over all \(t\).
\end{definition}

\begin{lemma}[Ensemble First-Moment Bound]
\label{lem:ensemble-first-moment}
For the odd-integer \(K\) randomized activation above, fix \(r\ge2\), \(\delta\in(0,1]\), \(2<c_1<c_2<2c_1\), \(\alpha\in(0,1]\), \(\kappa>0\), and \(t\in\{0,\ldots,n\}\) satisfying
\[
\frac{\delta}{c_2}<1-\frac{t}{n}<\frac{\delta}{c_1}.
\]
Then,
\[
\frac{1}{n}\log_2 \E\!\left[\mathsf{Coll}^{\mathrm{ens}}_{r,\delta}(t)\right]
\le
2 + 2(r-1)\left(1-\frac{t}{n}\right)
- \alpha\,\Lambda_{r,\delta}(\kappa,K)
+ o(1).
\]
\end{lemma}

\begin{proof}
We follow the same first-moment argument as in the proof of \Cref{thm:randomized-ogp-first-moment}.

Fix a candidate family $(x^{(j)},y^{(j)})_{j=1}^r$ satisfying the ensemble constraints. For each row of the matrices, define the Gaussian projections
\[
X_j = \langle a^{(j)}, x^{(j)} \rangle,\qquad
Y_j = \langle a^{(j)}, y^{(j)} \rangle,
\]
where the row vectors $a^{(j)}$ are jointly Gaussian according to the correlated ensemble.

The internal overlap condition implies
\[
d(X_j,Y_j)=\frac{1-\mathrm{Cov}(X_j,Y_j)}{2}\ge \delta/2.
\]
For distinct $j\neq j'$, the suffix rows of $a^{(j)}$ and $a^{(j')}$ are independent, and only the shared prefix contributes to the covariance. Since the family is $t$-prefix aligned,
\[
\mathrm{Cov}(X_j,X_{j'})=\frac{t}{n},
\qquad
\mathrm{Cov}(Y_j,Y_{j'})=\frac{t}{n}.
\]
This holds because $\veca^{(j)}$ and $\veca^{(j')}$ share their first $t$ columns, have independent suffixes, and prefix alignment gives $\vecx^{(j)}_{\le t}=\vecx^{(j')}_{\le t}$; therefore $\mathrm{Cov}(X_j,X_{j'})=(1/n)\,\vecx^{(j)\top}_{\le t}\,\vecx^{(j')}_{\le t}=t/n$, and similarly for $Y_j,Y_{j'}$.
Therefore
\[
d(X_j,X_{j'})+d(Y_j,Y_{j'})
=1-\frac{t}{n}.
\]
We now verify the separated subfamily needed in the row-wise arithmetization, using the same \(\eta_0\) fixed before the definition of \(\Lambda_{r,\delta}\). Fix a nonempty \(S\subseteq[r]\), write \(s=|S|\), and choose \(j_0\in S\). The variables \(\{X_j:j\in S\}\cup\{Y_{j_0}\}\) form a separated \((s+1)\)-subfamily. Indeed, for distinct \(j\neq j'\) in \(S\),
\[
d(X_j,X_{j'})=\frac{1-t/n}{2}\ge \frac{\delta}{2c_2},
\]
so their covariance separation is at least \(\delta/c_2\ge \eta_0\delta\). For \(j\neq j_0\), the covariance between \(X_j\) and \(Y_{j_0}\) comes only from the shared prefix, and prefix alignment gives \(x^{(j)}_{\le t}=x^{(j_0)}_{\le t}\):
\[
\mathrm{Cov}(X_j,Y_{j_0})
=
\frac{1}{n}\left\langle x^{(j_0)}_{\le t},y^{(j_0)}_{\le t}\right\rangle .
\]
If \(D_{\mathrm{pre}}=\Delta_{\mathrm{Ham}}(x^{(j_0)}_{\le t},y^{(j_0)}_{\le t})\), then
\[
d(X_j,Y_{j_0})
=
\frac{1-t/n}{2}+\frac{D_{\mathrm{pre}}}{n}.
\]
The internal-overlap condition gives
\[
D_{\mathrm{pre}}
\ge
\delta n/2-(n-t),
\]
and hence
\[
d(X_j,Y_{j_0})
\ge
\frac{\delta}{2}-\frac{1-t/n}{2}
>
\frac{\delta}{2}\left(1-\frac1{c_1}\right).
\]
Thus the covariance separation between \(X_j\) and \(Y_{j_0}\) is greater than \(\delta(1-1/c_1)\ge\eta_0\delta\), while \(d(X_{j_0},Y_{j_0})\ge\delta/2\) gives covariance separation at least \(\delta\). Therefore every pair in this \((s+1)\)-subfamily has covariance at most \(1-\eta_0\delta\).

	For the fixed row under consideration, all \(r\) ensemble instances use the same thresholds for that row. Thus the collision event on this row has the same arithmetization as in Lemma~\ref{lem:uniform-contraction}:
	\[
	\Pr[\varphi(X_j)=\varphi(Y_j)\ \forall j]
	=
	2^{-r}\sum_{S\subseteq[r]}
	\E\!\left[\prod_{j\in S}\varphi(X_j)\varphi(Y_j)\right].
	\]
	For each nonempty \(S\), the separated \((s+1)\)-subfamily exhibited above replaces the use of Lemma~\ref{lem:r+1 tuple exists}; applying Claim~\ref{claim:main rogp} with \(\eta=\eta_0\delta\) gives the same bound for the corresponding moment term. Therefore
	\[
	\Pr[\varphi(X_j)=\varphi(Y_j)\ \forall j] \le 2^{-\Lambda_{r,\delta}(\kappa,K)}.
	\]
Since the rows are independent, the probability that all $r$ pairs collide across all $m=\alpha n$ rows is at most,
\[
2^{-\alpha n\Lambda_{r,\delta}(\kappa,K)}\,.
\]
The number of $t$-prefix aligned families is at most
\[
2^{(2+2(r-1)(1-t/n))n}.
\]
Indeed, the first pair has at most $2^{2n}$ choices, and each later pair is fixed on the first $t$ coordinates and arbitrary on the remaining $n-t$ coordinates.

Finally, taking expectations and combining the bounds yields the stated exponent. The conclusion then follows by \Cref{lem:first-moment}.
\end{proof}

\begin{corollary}[Ensemble $r$-OGP]\label{cor:ensemble-rogp}
Fix \(\alpha\in(0,1]\), \(\kappa>0\), $\delta \in (0,1]$, $r \ge 2$, $2<c_1<c_2<2c_1$, and an odd integer \(K\ge1\).
Suppose that
\[
2 + 2(r-1)\frac{\delta}{c_1}
\;-\;
\alpha\,\Lambda_{r,\delta}(\kappa,K)
< 0.
\]
Then, for every $t \in \{0,1,\ldots,n\}$ satisfying
\[
\frac{\delta}{c_2}<1-\frac{t}{n}<\frac{\delta}{c_1},
\]
\[
\Pr\!\left[\mathsf{Coll}^{\mathrm{ens}}_{r,\delta}(t)\ge 1\right]
\le 2^{-\Omega(n)}.
\]

Consequently, ensemble $r$-OGP holds at $(\alpha,\delta)$ in the sense of \Cref{def:ensemble-rogp}.
\end{corollary}

\begin{proof}
For $t$ in the admissible window, Lemma~\ref{lem:ensemble-first-moment} gives the counting term
$2+2(r-1)(1-t/n)$, which is at most $2+2(r-1)\delta/c_1$.
Thus, for each fixed $t$ we have
$\mathbb{E}[\mathsf{Coll}^{\mathrm{ens}}_{r,\delta}(t)]\le 2^{-\Omega(n)}$ in the stated negative-exponent regime.
By \Cref{lem:first-moment},
$\Pr[\mathsf{Coll}^{\mathrm{ens}}_{r,\delta}(t)\ge 1]\le 2^{-\Omega(n)}$.
The exponent is uniform over all fixed \(t\) in the window, so this proves the stated ensemble \(r\)-OGP.
\end{proof}

\subsection{Hardness for Online Algorithms}\label{subsec:online-hardness}
An online algorithm observes the columns $\veca_1,\dots,\veca_n$ of $\matA$ sequentially and, after seeing the current column, must commit to the corresponding pair of input bits before seeing future columns. In the randomized activation model, the full threshold array $\mathcal T$ is auxiliary input available before the column stream begins.

Formally, after $\mathcal T$ is revealed, a randomized online algorithm $\mathcal A$ samples internal randomness $\omega$ at the start. At time $t$, it is given the threshold array, its coins, the observed column prefix, and its previous decisions, and outputs
\[
(x_t,y_t)
=
\mathcal A_t\!\left(
\mathcal T,\omega,
\veca_1,\ldots,\veca_t,
(x_1,y_1),\ldots,(x_{t-1},y_{t-1})
\right)
\in\{-1,1\}^2 .
\]
Equivalently, the algorithm may maintain an arbitrary internal state updated from the observed prefix, but it cannot inspect columns $\veca_{t+1},\ldots,\veca_n$ before committing to $(x_t,y_t)$.

For $t\in\{0,1,\ldots,n\}$, let $\mathsf{A}^{(t)}=(\matA^{(1)},\ldots,\matA^{(r)})$
be an ensemble whose first $t$ columns coincide and whose remaining $n-t$ columns are independent,
with the same thresholds shared across all $\matA^{(j)}$. 

\begin{theorem}\label{thm:online-hard}
Fix $2<c_1<c_2<2c_1$. Assume that ensemble $r$-OGP holds at parameters $(\alpha,\delta)$ with this forbidden-overlap window.
Then, over the random matrix, thresholds, and internal coins, every online algorithm $\mathcal{A}$ satisfies
\[
\Pr\!\left[
\begin{array}{c}
\mathcal{A}\text{ outputs }\vecx,\vecy\in\{\pm1\}^n
\text{ such that }
\varphi(\matA\vecx)=\varphi(\matA\vecy)\\
\text{and }
|\langle \vecx,\vecy\rangle|\le (1-\delta)n
\end{array}
\right]
\le 2^{-\Omega(n/r)}.
\]
Thus any constant choice of $r$ for which ensemble $r$-OGP holds gives an exponentially small online success probability.
\end{theorem}

\begin{proof}
Let $\xi=\frac13\delta(1/c_1-1/c_2)$. The interval below has positive constant width, so for all sufficiently large \(n\) we may choose \(t\in\{0,\ldots,n\}\) with constant slack in the ensemble-OGP window:
\[
\frac{\delta}{c_2}+\xi<1-\frac{t}{n}<\frac{\delta}{c_1}-\xi.
\]
Since $\delta,c_1,c_2$ are fixed constants, this choice has $n-t=\Theta(n)$.
Let $p$ denote the single-instance success probability of $\mathcal A$ for producing a $\delta$-extensive collision (unrelated to the anti-concentration parameter of the same name in \Cref{thm:randomized-ogp-first-moment}).
The coupling is designed so that online-ness forces the \(r\) outputs to agree on the shared prefix, while independent sign flips on the unseen suffix randomize the cross-overlaps without changing whether each output pair is a collision.
We use the following symmetrized coupled experiment. Sample the correlated ensemble $\mathsf A^{(t)}=(\matA^{(1)},\ldots,\matA^{(r)})$. For each $j\in[r]$ and each $i>t$, draw an independent sign $\sigma_i^{(j)}\in\{\pm1\}$ and let $\widetilde{\matA}^{(j)}$ be obtained from $\matA^{(j)}$ by replacing the suffix column $\veca_i^{(j)}$ with $\sigma_i^{(j)}\veca_i^{(j)}$. Run $\mathcal A$ on the matrices $\widetilde{\matA}^{(1)},\ldots,\widetilde{\matA}^{(r)}$, using the same internal coins and the same thresholds across all instances. Since the signed matrices have the same marginal law as the original matrices, each run succeeds with probability $p$.

Condition on the common prefix, the thresholds, and the internal coins. Given this information, the signed suffixes are independent across the $r$ instances and have identical conditional laws. Since the algorithm is online and the thresholds and coins are already fixed, the $r$ success events are conditionally independent and have the same conditional success probability, say $q$. Hence the probability that all $r$ runs succeed in the symmetrized coupled experiment is $\E\left[q^r\right]\ge(\E\left[q\right])^r=p^r$, by Jensen's inequality since $x\mapsto x^r$ is convex for $r\ge2$.

Let $(\vecu^{(j)},\vecv^{(j)})$ be the raw output of $\mathcal A$ on $\widetilde{\matA}^{(j)}$.
Since the first $t$ columns, the thresholds, and the internal coins are identical across the coupled runs, online-ness implies that on every outcome the raw prefixes
\[
\vecu^{(1)}_{\le t}=\cdots=\vecu^{(r)}_{\le t},
\qquad
\vecv^{(1)}_{\le t}=\cdots=\vecv^{(r)}_{\le t}
\]
are identical across all $j\in[r]$.

For $i>t$, set
\[
x_i^{(j)}=\sigma_i^{(j)}u_i^{(j)},\qquad
y_i^{(j)}=\sigma_i^{(j)}v_i^{(j)}.
\]
On the prefix, set $x_i^{(j)}=u_i^{(j)}$ and $y_i^{(j)}=v_i^{(j)}$. Then
\[
\matA^{(j)}\vecx^{(j)}=\widetilde{\matA}^{(j)}\vecu^{(j)},\qquad
\matA^{(j)}\vecy^{(j)}=\widetilde{\matA}^{(j)}\vecv^{(j)},
\]
where $\widetilde{\matA}^{(j)}$ is the signed matrix seen by the algorithm. Thus this post-processing preserves collision equations and internal overlaps: since $(\sigma_i^{(j)})^2=1$ for all $i$, the inner product $\langle\vecx^{(j)},\vecy^{(j)}\rangle=\sum_i (\sigma_i^{(j)})^2 u_i^{(j)} v_i^{(j)}=\langle\vecu^{(j)},\vecv^{(j)}\rangle$ is unchanged.

By Gaussian sign symmetry, the signs $\sigma_i^{(j)}$ are independent of the signed matrices $\widetilde{\matA}^{(j)}$. The raw outputs are functions of the signed matrices, the thresholds, and the fixed coins, so after conditioning on the signed matrices and raw outputs the suffix flip signs remain independent. For every distinct $j,j'$, the suffix products
\[
x_i^{(j)}x_i^{(j')}=\sigma_i^{(j)}\sigma_i^{(j')}u_i^{(j)}u_i^{(j')}
\]
are independent mean-zero signs over $i>t$, and the same holds for the $y$-vectors. Therefore, by \Cref{lem:hoeffding} with $\lambda=\xi n/(4(n-t))$ and a union bound over pairs $j\neq j'$,
\[
\left|\frac{1}{n}\langle \vecx^{(j)},\vecx^{(j')}\rangle-\frac{t}{n}\right|\le \xi/4,
\qquad
\left|\frac{1}{n}\langle \vecy^{(j)},\vecy^{(j')}\rangle-\frac{t}{n}\right|\le \xi/4
\]
for all distinct $j,j'$ except with probability $2^{-\Omega(n)}$.
Since $1-t/n$ has slack $\xi$ in the admissible window, these inequalities imply
\[
1-\frac{\delta}{c_1}
\le
\frac{\langle \vecx^{(j)},\vecx^{(j')}\rangle+\langle \vecy^{(j)},\vecy^{(j')}\rangle}{2n}
\le
1-\frac{\delta}{c_2}
\]
for all distinct $j,j'$.

On the event that all $r$ runs succeed and the above concentration event holds, the outputs form a $t$-prefix aligned, ensemble $\delta$-admissible family of collisions. Hence
\[
p^r
\le
\Pr[\mathsf{Coll}^{\mathrm{ens}}_{r,\delta}(t)\ge1]
+2^{-\Omega(n)}
\le 2^{-\Omega(n)},
\]
where the last inequality is ensemble $r$-OGP. Therefore $p\le 2^{-\Omega(n/r)}$.
\end{proof}

\begin{corollary}[Online Collision Resistance]\label{cor:online-crh}
Fix $\alpha\in(0,1)$, \(\kappa>0\), $r\ge2$, $\delta\in(0,1]$, $2<c_1<c_2<2c_1$, and an odd integer \(K\ge1\). If
\[
2+2(r-1)\frac{\delta}{c_1}-\alpha\Lambda_{r,\delta}(\kappa,K)<0,
\]
then every online algorithm has success probability at most $2^{-\Omega(n/r)}$ for outputting a $\delta$-extensive collision. In particular, any constant choice of $r$ satisfying the condition gives online $\delta$-extensive collision resistance at density $\alpha$. (The ensemble condition uses $\delta/c_1$ in place of the entropy $H(\delta/(2c_1))$ appearing in \Cref{corollary:rogp}; this is because prefix alignment restricts the external-overlap count to the suffix, contributing $(r-1)(1-t/n)\le(r-1)\delta/c_1$ linearly rather than entropically.)
\end{corollary}
\begin{proof}
The condition implies ensemble $r$-OGP by \Cref{cor:ensemble-rogp}. The conclusion follows from \Cref{thm:online-hard}.
\end{proof}

\bibliographystyle{splncs04}
\bibliography{refs}

@article{BILSZ15,
  author  = {Baldassi, C. and Ingrosso, A. and Lucibello, C. and Saglietti, L. and Zecchina, R.},
  title   = {{Subdominant Dense Clusters Allow for Simple Learning and High Computational Performance in Neural Networks with Discrete Synapses}},
  journal = {Physical Review Letters},
  volume  = {115},
  pages   = {128101},
  year    = {2015},
  doi     = {10.1103/PhysRevLett.115.128101}
}

@article{BS20,
  author    = {Bansal, N. and Spencer, J. H.},
  title     = {{On-Line Balancing of Random Inputs}},
  journal   = {Random Structures \& Algorithms},
  volume    = {57},
  number    = {4},
  pages     = {879--891},
  year      = {2020},
  publisher = {Wiley},
  doi       = {10.1002/rsa.20955}
}

@article{GamarnikJagannathAOP21,
  author  = {Gamarnik, D. and Jagannath, A.},
  title   = {{The Overlap Gap Property and Approximate Message Passing Algorithms for $p$-Spin Models}},
  journal = {The Annals of Probability},
  volume  = {49},
  number  = {1},
  pages   = {180--205},
  year    = {2021},
  doi     = {10.1214/20-AOP1448}
}

@inproceedings{GJW20,
  author    = {Gamarnik, D. and Jagannath, A. and Wein, A.},
  title     = {{Low-Degree Hardness of Random Optimization Problems}},
  booktitle = {IEEE Symposium on Foundations of Computer Science (FOCS)},
  pages     = {131--140},
  year      = {2020},
  publisher = {IEEE},
  doi       = {10.1109/FOCS46700.2020.00021}
}

@inproceedings{LiZM19,
  author    = {Li, K. and Zhang, T. and Malik, J.},
  title     = {{Approximate Feature Collisions in Neural Nets}},
  booktitle = {Advances in Neural Information Processing Systems (NeurIPS)},
  year      = {2019},
  pages     = {15816--15824}
}

@article{rosenblatt1958perceptron,
  author  = {Rosenblatt, F.},
  title   = {{The Perceptron: A Probabilistic Model for Information Storage and Organization in the Brain}},
  journal = {Psychological Review},
  volume  = {65},
  number  = {6},
  pages   = {386--408},
  year    = {1958},
  doi     = {10.1037/h0042519}
}

@book{Rosenblatt239697,
  author    = {Rosenblatt, F.},
  title     = {{Principles of Neurodynamics: Perceptrons and the Theory of Brain Mechanisms}},
  publisher = {Spartan},
  address   = {Washington, DC},
  year      = {1962}
}

@article{mcculloch1943logical,
  author  = {McCulloch, W. S. and Pitts, W.},
  title   = {{A Logical Calculus of the Ideas Immanent in Nervous Activity}},
  journal = {The Bulletin of Mathematical Biophysics},
  volume  = {5},
  number  = {4},
  pages   = {115--133},
  year    = {1943},
  doi     = {10.1007/BF02478259}
}

@article{SubagInvent17,
  author  = {Subag, E.},
  title   = {{The Geometry of the Gibbs Measure of Pure Spherical Spin Glasses}},
  journal = {Inventiones Mathematicae},
  volume  = {210},
  number  = {1},
  pages   = {135--209},
  year    = {2017},
  doi     = {10.1007/s00222-017-0726-4}
}

@inproceedings{GKPX22,
  author    = {Gamarnik, D. and K{\i}z{\i}lda\u{g}, E. C. and Perkins, W. and Xu, C.},
  title     = {{Algorithms and Barriers in the Symmetric Binary Perceptron Model}},
  booktitle = {IEEE Symposium on Foundations of Computer Science (FOCS)},
  pages     = {576--587},
  year      = {2022},
  publisher = {IEEE},
  doi       = {10.1109/FOCS54457.2022.00061}
}

@article{GamarnikSudanAOP17,
  author    = {Gamarnik, D. and Sudan, M.},
  title     = {{Limits of Local Algorithms over Sparse Random Graphs}},
  journal   = {The Annals of Probability},
  volume    = {45},
  number    = {4},
  pages     = {2353--2376},
  year      = {2017},
  publisher = {Institute of Mathematical Statistics},
  doi       = {10.1214/16-AOP1114}
}

@article{Gardner88,
  author    = {Gardner, E.},
  title     = {{The Space of Interactions in Neural Network Models}},
  journal   = {Journal of Physics A: Mathematical and General},
  volume    = {21},
  number    = {1},
  pages     = {257--270},
  year      = {1988},
  publisher = {IOP Publishing},
  doi       = {10.1088/0305-4470/21/1/030}
}

@article{GardnerDerrida88,
  author    = {Gardner, E. and Derrida, B.},
  title     = {{Three Unfinished Works on the Optimal Storage Capacity of Networks}},
  journal   = {Journal of Physics A: Mathematical and General},
  volume    = {22},
  number    = {12},
  pages     = {1983--1994},
  year      = {1989},
  publisher = {IOP Publishing},
  doi       = {10.1088/0305-4470/22/12/004}
}

@inproceedings{KimR95,
  author    = {Kim, J. H. and Roche, J. R.},
  title     = {{On the Optimal Capacity of Binary Neural Networks: Rigorous Combinatorial Approaches}},
  booktitle = {Proceedings of the Conference on Computational Learning Theory (COLT)},
  pages     = {240--249},
  year      = {1995},
  publisher = {ACM},
  doi       = {10.1145/225298.225327}
}

@book{MezardMontanariBook,
  author    = {M{\'e}zard, M. and Montanari, A.},
  title     = {{Information, Physics, and Computation}},
  year      = {2009},
  publisher = {Oxford University Press},
  doi       = {10.1093/acprof:oso/9780198570837.001.0001}
}

@article{swp,
  author    = {Benedetti, M. and Bogdanov, A. and Malatesta, E. and M{\'e}zard, M. and Perrupato, G. and Rosen, A. and Schwartzbach, N. I. and Zecchina, R.},
  title     = {{Overlap Gap and Computational Thresholds in the Square Wave Perceptron}},
  journal   = {Journal of Statistical Mechanics: Theory and Experiment},
  volume    = {2025},
  number    = {12},
  pages     = {123303},
  year      = {2025},
  doi       = {10.1088/1742-5468/ae23be},
  eprint    = {2506.05197},
  eprinttype= {arXiv},
}

@misc{collisionResistance2025,
  author    = {Benedetti, M. and Bogdanov, A. and Malatesta, E. and M{\'e}zard, M. and Perrupato, G. and Rosen, A. and Schwartzbach, N. I. and Zecchina, R.},
  title     = {{Are Neural Networks Collision Resistant?}},
  year      = {2025},
  eprint    = {2509.20262},
  eprinttype= {arXiv},
  url       = {https://eprint.iacr.org/2025/1756},
  note      = {Cryptology ePrint Archive, Paper 2025/1756}
}

@article{Gamarnik21OGPsurvey,
  author      = {Gamarnik, D.},
  title       = {{The Overlap Gap Property: A Topological Barrier to Optimizing over Random Structures}},
  journal     = {Proceedings of the National Academy of Sciences},
  volume      = {118},
  number      = {41},
  pages       = {e2108492118},
  year        = {2021},
  doi         = {10.1073/pnas.2108492118}
}

@article{OzbulakEtAlSciRep26,
  author      = {Ozbulak, U. and Rao, S. and De Neve, W. and Vankerschaver, J. and Van Messem, A. and Gasparyan, M.},
  title       = {{Exact Feature Collisions in Neural Networks}},
  journal     = {Scientific Reports},
  volume      = {16},
  pages       = {10139},
  year        = {2026},
  doi         = {10.1038/s41598-026-40605-4}
}

@inproceedings{VV25,
  author    = {Vafa, N. and Vaikuntanathan, V.},
  title     = {{Symmetric Perceptrons, Number Partitioning and Lattices}},
  booktitle = {Proceedings of the ACM Symposium on Theory of Computing (STOC)},
  pages     = {2191--2202},
  year      = {2025},
  publisher = {ACM},
  doi       = {10.1145/3717823.3718263}
}

@inproceedings{10.1145/3717823.3718124,
author = {Li, Shuangping and Schramm, Tselil and Zhou, Kangjie},
title = {Discrepancy Algorithms for the Binary Perceptron},
year = {2025},
isbn = {9798400715105},
booktitle = {Proceedings of the 57th Annual ACM Symposium on Theory of Computing},
pages = {1668--1679},
numpages = {12},
location = {Prague, Czechia},
series = {STOC '25},
doi = {10.1145/3717823.3718124}
}

@book{Vershynin2018,
  author    = {Vershynin, Roman},
  title     = {{High-Dimensional Probability: An Introduction with Applications in Data Science}},
  publisher = {Cambridge University Press},
  year      = {2018},
  doi       = {10.1017/9781108231596}
}

@book{DubhashiPanconesi2009,
  author    = {Dubhashi, Devdatt P. and Panconesi, Alessandro},
  title     = {{Concentration of Measure for the Analysis of Randomized Algorithms}},
  publisher = {Cambridge University Press},
  year      = {2009},
  doi       = {10.1017/CBO9780511581274}
}

\end{document}